\documentclass[journal]{IEEEtran}
\usepackage{indentfirst}
\usepackage{color}
\usepackage{amsfonts,amsmath,amssymb,amsthm}
\usepackage{graphicx,multirow,bm}
\usepackage{cite}
\usepackage{arydshln}

\newtheorem{Theorem}{Theorem}
\newtheorem{Example}{Example}
\newtheorem{Remark}{Remark}
\newtheorem{Lemma}{Lemma}
\newtheorem{Construction}{Construction}
\newtheorem{Proposition}{Proposition}
\newtheorem{Fact}{Fact}

\newenvironment{psmallmatrix}
{\left(\begin{smallmatrix}}
{\end{smallmatrix}\right)}
\begin{document}
\title{MDS Array Codes With Small Sub-packetization Levels and Small Repair Degrees
}

\author{Jie~Li,~\IEEEmembership{Senior~Member,~IEEE,} Yi~Liu,~\IEEEmembership{Member,~IEEE,}
Xiaohu~Tang,~\IEEEmembership{Fellow,~IEEE,}\\
Yunghsiang S.~Han,~\IEEEmembership{Fellow,~IEEE,} Bo~Bai,~\IEEEmembership{Senior~Member,~IEEE,} and Gong Zhang
\thanks{The work of X. Tang was supported by the National
Natural Science Foundation of China under Grant 62371401.}
\thanks{J. Li was with the Theory Lab, Central Research Institute, 2012 Labs, Huawei Technologies Co., Ltd., Shatin, New Territories, Hong Kong SAR, China, and is with the Department of Computer Science, City University of Hong Kong, Hong Kong SAR  (e-mail: jieli873@gmail.com).}
\thanks{Y. Liu was with the Information Coding and Transmission Key Lab of Sichuan Province, CSNMT Int. Coop. Res. Centre (MoST), Southwest Jiaotong University, Chengdu, 610031, China (e-mail: yiliu.swjtu@outlook.com).}
\thanks{X. Tang is with the Information Coding and Transmission Key Lab of Sichuan Province, CSNMT Int. Coop. Res. Centre (MoST), Southwest Jiaotong University, Chengdu, 610031, China (e-mail: xhutang@swjtu.edu.cn).}
\thanks{Yunghsiang S.~Han is with the School of Computing and Information Technology, Great Bay University, Dongguan 523000, China
 (e-mail: yunghsiangh@gmail.com).}
 \thanks{B. Bai is with the Theory Lab, Central Research Institute, 2012 Labs, Huawei Technologies Co., Ltd., Shatin, New Territories, Hong Kong SAR, China (e-mail: baibo8@huawei.com).}
\thanks{G. Zhang is with the Theory Lab, Central Research Institute, 2012 Labs, Huawei Technologies Co., Ltd., Shenzhen, China (e-mail: nicholas.zhang@huawei.com).}
 }

\maketitle

\begin{abstract}
High-rate minimum storage regenerating (MSR) codes are known to require a large sub-packetization level, which can make metadata management difficult and hinder implementation in practical storage systems. A few maximum distance separable (MDS) array code constructions have been proposed to achieve a significantly smaller sub-packetization level by sacrificing a portion of the repair bandwidth. However, only a few constructions support repair degree (the number of contacted surviving nodes during the repair of a failed node) $d<n-1$. 
In this paper, we construct a new $(n, k)$ MDS array code with repair degree $d<n-1$, small sub-packetization level, and $(1+\epsilon)$-optimal repair bandwidth (i.e., $1+\epsilon$ times the optimal value). In addition, we apply a generic transformation to transform two specific $(\overline{n}, \overline{k})$ MSR codes with repair degree $\overline{d}<\overline{n}-1$ into another two $(n=s\overline{n},k)$ MDS array codes with repair degree $d<n-1$, a small sub-packetization level, and $(1+\epsilon)$-optimal repair bandwidth. As in the one constructed by Guruswami \textit{et al.}, all the new MDS array codes have to contact some compulsory surviving nodes during the repair process. However, in contrast to the former, the new MDS array codes are applicable for both small and large code parameters. Comparisons further show that the new codes have smaller field sizes and smaller sub-packetization levels, and thus are more flexible.
\end{abstract}

\begin{IEEEkeywords}
Maximum distance separable, minimum storage regenerating codes, repair bandwidth, repair degree, sub-packetization.
\end{IEEEkeywords}

\section{Introduction}
\IEEEPARstart{I}{n} distributed storage systems, data is stored across multiple unreliable storage nodes. Therefore, redundancy must be introduced to provide fault tolerance and maintain storage reliability. Compared to replication, erasure coding is a more efficient redundancy mechanism. A typical example of erasure codes is scalar MDS codes, which can encode a file of $k$ symbols over a finite field $\mathbf{F}_q$ into a codeword of $n$ symbols. These $n$ symbols are dispersed across a distributed storage system of $n$ storage nodes so that any $k$ out of the $n$ nodes can recover the original file. However, when repairing a failed node, the conventional reconstruction-based repair method
downloads the contents of $k$ surviving nodes to regenerate the failed node. This leads to excessive \textit{repair bandwidth}, which is defined as the amount of data downloaded to repair a failed node.

MDS array codes are often preferred in distributed storage systems as they can reduce the repair bandwidth. In the typical scenario mentioned above, the codeword is an array of size $N\times n$ instead of a vector, and each node stores one
length-$N$ column vector from the codeword. The parameter $N$ is referred to as the \textit{sub-packetization level}. According to the cut-set bound in \cite{dimakis2010network}, the repair bandwidth of $(n,k)$ MDS array codes with a sub-packetization level $N$ is lower bounded by
\begin{equation}\label{Eqn_gamma_opt}
\gamma_{\rm optimal}=\frac{d}{d-k+1}N,
\end{equation}
where $d$ ($k\le d<n$) denotes the number of helper nodes contacted during the repair process and is called the \textit{repair degree}. MDS array codes with repair bandwidth attaining this lower bound are said to have the \textit{optimal repair bandwidth} and are also referred to as MSR codes in \cite{dimakis2010network}.

The construction of MSR codes has received considerable
attention, leading to significant progress over the past decade \cite{Goparaju,wang2012long,invariant subspace,hadamard,tian2014characterizing,Hadamard strategy,Long IT,Barg1,Barg2,Zigzag,Sasidharan-Kumar2,Rashmi2011optimal,han2015update,li2017generic,li2018generic,hou2020binary,elyasi2020cascade,li2022generic,liu2022generic,li2022constructing,zhang2023vertical}. Existing high-rate MSR code constructions require a significantly large sub-packetization level $N$. This is necessary as Alrabiah and Guruswami recently showed that the lower
bound on $N$ is exponential in $k$ for fixed $r=n-k$, and they conjectured
that the construction in \cite{Long IT} with $N=r^{\frac{n}{r+1}}$ is exactly tight \cite{alrabiah2021exponential}. However, an MDS array code with a large sub-packetization level can limit the design space regarding various system parameters and make metadata management challenging, hindering implementation in practical systems \cite{Rawat}.

Based on existing constructions and theoretical bounds, it has been deemed impossible to construct high-rate MSR codes with a small sub-packetization level. Nevertheless, recent work demonstrates that high-rate MDS array codes with a small sub-packetization level can be constructed by sacrificing the optimality of the repair bandwidth. In \cite{Rawat}, two high-rate MDS array codes with small sub-packetization levels and $(1+\epsilon)$-optimal repair bandwidth were presented. Furthermore, in \cite{eMSR_d_eq_n-1}, a generic transformation was proposed, which can convert any MSR code into an MDS array code with a small sub-packetization level and $(1+\epsilon)$-optimal repair bandwidth. However, all the resultant MDS array codes in \cite{Rawat,eMSR_d_eq_n-1} only work for repair degree $d=n-1$, meaning that all surviving nodes need to be contacted to repair a failed node.

Recently, Guruswami \textit{et al.} generalized the work in \cite{Rawat} to support $d<n-1$ \cite{Guruswami2020}, but where $d$ helper nodes must include some compulsory surviving nodes. 
Although this constraint may cause node bottlenecks, such bottlenecks can be amortized under the assumption of random node failures.
The MDS array code in \cite{Guruswami2020} is derived by combining an MSR code and a scalar linear code with specific parameters. For the sake of convenience, we will refer to the MDS array code in \cite{Guruswami2020} as the GLJ code in this paper. The GLJ code is certainly of theoretical interest. However, it does not offer immediate practical value, as the construction only applies to large code parameters. For instance, when $r=3$, to achieve a repair bandwidth $\gamma\le 1.7\gamma_{\rm optimal}$, where $\gamma_{\rm optimal}$ is defined in \eqref{Eqn_gamma_opt}, the GLJ code requires a code length $n=961^{12}$ and sub-packetization level $N=90\times 2^{961}$, as outlined in Tables I and II in \cite{Guruswami2020}.  Very recently, the work in \cite{Guruswami2020} was further extended in \cite{Ramkumar2025MSR} to remove the compulsory
nodes during the repair process by using an auxiliary code with prescribed distance. 
It also provides asymptotic existence results and explicit families based on algebraic-geometry
codes.
Another way to construct MDS array codes with small repair bandwidth, small sub-packetization level, and small repair degree is via the piggybacking framework proposed in \cite{rashmi2014hitchhiker} and \cite{Rashmi2017piggybacking}. Interested readers could refer to \cite{rashmi2014hitchhiker}, \cite{Rashmi2017piggybacking}, and their follow-up works.


In this paper, for any integers $m,s\ge2$, we first construct an $(n=2ms, k)$ MDS array code $\mathcal{C}_1$ with $(1+\epsilon)$-optimal repair bandwidth, where the repair degree $d$ satisfies $\max\{s-1,k+1\}\le d<n-1$ and the sub-packetization level is as small as $N=(d-k+1)^{m}$. In addition, we discover that the generic transformation in \cite{eMSR_d_eq_n-1}, which was originally used to convert MSR codes with $d=n-1$ to MDS array codes with small sub-packetization levels, can also be applied to some MSR codes with $d<n-1$.
By applying this transformation to two well-known MSR codes in Constructions 2 and 5 in \cite{Barg1} (referred to as YB code 1 and YB code 2, respectively), we construct two MDS array codes $\mathcal{C}_2$ and $\mathcal{C}_3$ with small sub-packetization levels and $(1+\epsilon)$-optimal repair bandwidth. Note that all the new array codes support a repair degree of $d<n-1$ by including some compulsory surviving nodes as helper nodes, and can be constructed over small finite fields.

By applying the generic transformation in \cite{eMSR_d_eq_n-1} to the YB code 1 and YB code 2 with $d<n-1$ in \cite{Barg1}, the proposed MDS array codes $\mathcal{C}_2$ and $\mathcal{C}_3$  in this paper have similar structures to those in \cite{eMSR_d_eq_n-1} that focused on $d=n-1$. The key difference lies in that we need to cancel additional interference during the node repair process of the codes with $d<n-1$,
since there are some unconnected nodes during the node repair process. In contrast, all surviving nodes are connected for codes in \cite{eMSR_d_eq_n-1}. The presence of this extra interference, which needs to be canceled, poses greater challenges to this work. To cancel the 
interference from the unconnected nodes, we elaborately choose scaling coefficients for the blocks of the
parity-check matrix, such that the 
interference from the unconnected nodes can be determined according to that from the connected nodes while simultaneously guaranteeing the MDS property. More details will be provided in later sections. 

Compared to the earlier work \cite{eMSR_d_eq_n-1}, the proposed new MDS array codes with $d<n-1$ have the advantage that they can tolerate up to $n-1-d$ other nodes being temporarily unavailable during a repair. That is, even if some surviving nodes are offline, the repair process can still proceed, significantly enhancing system robustness and data durability. Compared with the MSR codes with $d<n-1$ in \cite{Barg1}, the proposed MDS array codes have significantly lower sub-packetization levels, i.e., $w^m$ and $w^{\bar{n}}$ versus $w^n$, where $m\ge2$ and $\bar{n}\ge5$.
Compared to the GLJ code in \cite{Guruswami2020}, the new $(n, k)$ MDS array codes with a repair degree $d<n-1$ have several advantages:
\begin{itemize}
\item The new MDS array codes are valid for any code length $n\ge 8$, whereas the GLJ code in \cite{Guruswami2020} only works for very large code lengths \footnote{Although a short code can be obtained by shortening a longer code, this typically requires a field size as large as that of the longer code.}.
\item
The new MDS array codes have an extremely small sub-packetization level $N$, which can even be a small constant independent of the code length $n$. More specifically, the smallest sub-packetization level is $2^2$ for the new code $\mathcal{C}_1$ and $2^5$ for both new codes $\mathcal{C}_2$ and $\mathcal{C}_3$. In contrast, the GLJ code in \cite{Guruswami2020} only supports a very large sub-packetization level $N$.

\item The required finite field sizes of all the new MDS array codes are comparable to the code lengths. For example, when $s=m=w=2$ and $\bar{n}=5$, the code lengths of the new codes $\mathcal{C}_1$, $\mathcal{C}_2$ and $\mathcal{C}_3$ are $8,10$, and $10$, respectively. Field sizes of 9, 11, and 13 are respectively suffice for these examples, which are relatively much smaller than those of the GLJ code in \cite{Guruswami2020} (e.g., $q>2883n$ when $n=961^{12}$ and $r=3$), making them easier to implement in practical systems.
\item The new code $\mathcal{C}_1$ has the smallest sub-packetization level among all the new MDS array codes, $\mathcal{C}_2$ has the optimal update property, and $\mathcal{C}_3$ has a low access property.
\end{itemize}


The remainder of the paper is organized as follows.
Section II reviews some necessary preliminaries of high-rate MDS array codes. Section III proposes a new $(n, k)$ MDS array code with repair degree $d<n-1$ and small sub-packetization level. Section IV constructs two explicit MDS array codes with small sub-packetization levels and $(1+\epsilon)$-optimal repair bandwidth by applying the generic transformation in \cite{eMSR_d_eq_n-1} to two well-known codes in \cite{Barg1}. Section \ref{sec:comp} compares key parameters
among the MDS array codes proposed in this paper and some
existing ones, i.e., the MSR code in \cite{li2023MDS}, YB codes 1, 2 in \cite{Barg1}, GLJ code in \cite{Guruswami2020}. Finally, Section \ref{sec:conclusion} concludes
the study.

\section{Preliminaries}
In this section, we introduce some preliminaries on MDS array codes and a special partition for a given basis. This paper denotes by $q$ a prime power and $\mathbf{F}_q$ the finite field with $q$ elements. Let $[a: b)$ represent the set $\{a, a+1, \ldots, b-1\}$ for two integers $a$ and $b$ with $a<b$. 

\subsection{MDS Array Codes}

An $(n,k)$ array code with sub-packetization level $N$ can be defined by an $rN\times n N$ matrix $A$ of full rank over $\mathbf{F}_q$ whose codewords satisfy
\begin{equation}\label{Eqn parity check eq}
A
\left(\begin{array}{c}
\mathbf{f}_0\\
\mathbf{f}_1\\
\vdots\\
\mathbf{f}_{n-1}
\end{array}
\right)=\mathbf{0}_{rN},
\end{equation}
where $r=n-k$, $\mathbf{f}_i$ is a column vector of length $N$ over $\mathbf{F}_q$, $\mathbf{0}_{rN}$ denotes the zero column vector of length $rN$ and will be abbreviated as $\mathbf{0}$ if its length is not ambiguous. The $rN\times nN$ block matrix $A$ in \eqref{Eqn parity check eq} is referred to as the \textit{parity-check matrix} of the code and can be represented as a block matrix
\begin{align*}
A&=(A_{t,i})_{t\in [0: r),i\in[0: n)}\\
&=\left(
\begin{array}{cccc}
A_{0,0} & A_{0,1} & \cdots & A_{0,n-1} \\
A_{1,0} & A_{1,1} & \cdots & A_{1,n-1} \\
\vdots & \vdots & \ddots & \vdots \\
A_{r-1,0} & A_{r-1,1} & \cdots & A_{r-1,n-1}
\end{array}\right),
\end{align*}
where $A_{t,j}$ is a square matrix of order $N$.
Note that for each $t\in[0: r)$, $\sum\limits_{i=0}^{n-1}A_{t,i}\mathbf{f}_i=0$ in \eqref{Eqn parity check eq} yields $N$ equations. 

An $(n,k)$ array code defined by \eqref{Eqn parity check eq} is MDS if any $r\times r$ sub-block matrix of the matrix $(A_{t,i})_{t\in [0: r),i\in[0: n)}$
is non-singular \cite{Barg1}. 
In particular, if
\begin{eqnarray}\label{Eqn A power}
A_{t,i}=A_i^{t}, ~t\in [0: r), ~i\in [0: n)
\end{eqnarray}
for some matrices $A_i$ of order $N$, then the MDS property of an array code with the parity-check matrix $A=(A_{t,i})_{t\in [0: r),i\in[0: n)}$ can be easily verified by the following lemma.

\begin{Lemma}\label{Lemma block matrix}
(Block Vandermonde matrix \cite{Barg1}) Let $B_0, B_1, \ldots, B_{v-1}$ be $N\times N$ matrices such that $B_iB_j=B_jB_i$ for all $i,j \in [0: v)$, then the matrix
\begin{equation*}
\left(
\begin{array}{cccc}
I & I & \cdots & I \\
B_0 & B_1 & \cdots & B_{v-1} \\
\vdots & \vdots & \ddots & \vdots \\
B_0^{v-1} & B_1^{v-1} & \cdots & B_{v-1}^{v-1} \\
\end{array}
\right)
\end{equation*}
is non-singular
if and only if $B_i-B_j$ is non-singular for all $i,j\in [0: v)$ with $i\ne j$.
\end{Lemma}




%

An $(n,k)$ array code with sub-packetization level $N$ can be deployed in distributed storage systems as follows. A file containing $kN$ symbols in a finite field $\mathbf{F}_q$ is encoded by an $(n,k)$  array code, a codeword consists of $n$ vectors $\mathbf{f}_i$, $i\in [0: n)$, each of length $N$, which are dispersed across a distributed storage system comprising $n$ individual storage nodes.
In certain applications, an $(n,k)$ MDS array code is desired to possess the \textit{optimal update property}, meaning that the
minimum possible number of symbols needs to
be updated when an information symbol is changed. In \cite{Barg1},
Ye and Barg showed that an $(n,k)$ MDS array code defined in the
form of \eqref{Eqn parity check eq} has the optimal update property if all the
block entries $A_{t, i}$ in the parity-check matrix $(A_{t,i})_{t\in [0: r),i\in[0: n)}$ are diagonal. Meanwhile, for those defined in the generator matrix form, they have the optimal update property if every
block entry in the generator matrix is either diagonal or a generalized permutation matrix according to \cite{Zigzag}. 

\subsection{Repair Mechanism}\label{sec:repair}
In this paper, we consider the exact repair of a single failed node.
A repair mechanism for $(n,k)$ array codes with a repair degree of $d=n-1$ was formulated in \cite{eMSR_d_eq_n-1}. In this subsection, we generalize the result to $d\le n-1$.

For an $(n,k)$ array code, suppose that node $i$ ($i\in[0: n)$) fails, applying a common selection matrix to the parity-check groups in \eqref{Eqn parity check eq} gives the following equations.

\begin{align}\label{Eqn repairment gene}
\nonumber \underbrace{\left(\begin{array}{c}
S_{i} A_{0,i} \\
S_{i} A_{1,i}\\
\vdots\\
S_{i} A_{r-1,i}
\end{array}\right)\mathbf{f}_i}_{\mathrm{useful ~data}}+\sum_{l\in L_i}\underbrace{\left(\begin{array}{c}
S_{i} A_{0,l}\\
S_{i} A_{1,l}\\
\vdots\\
S_{i} A_{r-1,l}
\end{array}\right)\mathbf{f}_l}_{\mathrm{unknown ~data~by~}\mathbf{f}_l}&
\\
+\sum_{j\in H_i}\underbrace{\left(\begin{array}{c}
S_{i} A_{0,j}\\
S_{i} A_{1,j}\\
\vdots\\
S_{i} A_{r-1,j}
\end{array}\right)\mathbf{f}_j}_{\mathrm{interference ~by~}\mathbf{f}_j}=&\bf{0},
\end{align}
where $S_{i}$ is an $\frac{N}{d-k+1}\times N$ full rank matrix and is called the \textit{selection matrix}, $H_i$ denotes the set of indices of the helper nodes, i.e., a $d$-subset of $[0: n)\backslash\{i\}$, and $L_i=[0: n)\backslash (H_i\cup\{i\})$
denotes the set of indices of unconnected nodes. We denote the data downloaded from helper node $j$
as $R_{i,j}\mathbf{f}_j$, where $R_{i,j}$ is an $\beta_{i,j} \times N$ matrix of full rank with $\beta_{i,j}\le N$, and is called the \textit{repair matrix} used by helper node $j$ to repair node $i$.

In \eqref{Eqn repairment gene}, the data at the failed node $i$ and the unconnected nodes in $L_i$ are treated as unknowns. Then, node $i$ can be regenerated if the following three conditions are satisfied.
\begin{itemize}
\item [R1.] The coefficient matrix of the system of linear equations (w.r.t. the variable $\mathbf{f}_i$) in \eqref{Eqn repairment gene} has rank $N$, i.e.,\begin{equation}\label{repair_node_requirement1 d}
\mbox{\rm rank}\left(\begin{array}{c}
S_{i} A_{0,i}\\
S_{i} A_{1,i}\\
\vdots\\
S_{i} A_{r-1,i}
\end{array}\right)=N,
\end{equation}
which guarantees unique recovery of $\mathbf{f}_i$ once the remaining terms in \eqref{Eqn repairment gene} are known.

\item [R2.] The interference caused by $\mathbf{f}_j$ can be determined from the downloaded data $R_{i,j}\mathbf{f}_{j}$ for $j\in H_i$, and there exists some matrix $R_{i,l}$ such that the unknown data related to $\mathbf{f}_l$ can be determined from $R_{i,l}\mathbf{f}_{l}$ for $l\in L_i$, i.e.,
\begin{align}\label{repair_node_requirement3 d}
\textrm{rank} \begin{pmatrix}    
R_{i,j} \\
S_{i} A_{t,j}
\end{pmatrix}=\mbox{\rm rank}(R_{i,j}),~j\in H_i\cup L_i,~t\in [0: r).
\end{align}
R2 helps to cancel the interference related to the connected nodes from $H_i$ and characterize the amount of interference related to
the unconnected nodes from $L_i$. For the sake of simplicity, we also refer to $R_{i,l}$ as the \textit{repair matrix} used by node $l$ to repair node $i$ for $l\in L_i$. 

\item [R3.] The data $R_{i,l}\mathbf{f}_l$, $l\in L_i$ can be determined from the downloaded data $R_{i,j}\mathbf{f}_j$, $j\in H_i$, i.e., there is a function $\phi_{i,l}$ such that
\begin{equation*}
R_{i,l}\mathbf{f}_l=\phi_{i,l}\left(R_{i,j}\mathbf{f}_j, j\in H_i\right).
\end{equation*}
R3 requires canceling the interference $\mathbf{f}_l, l\in L_i$ from the unconnected nodes by utilizing the downloaded data $R_{i,j}\mathbf{f}_j, j\in H_i$.
\end{itemize}

Alternatively, if \eqref{repair_node_requirement3 d} is satisfied, \eqref{Eqn repairment gene} can be reformulated as
\begin{align}\label{eqn_alternatively_repair}
\nonumber&
\begin{pmatrix}
S_{i} A_{0,i} & B_{0, l_0, i}&\cdots & B_{0, l_{n-d-2}, i}\\
S_{i} A_{1,i} & B_{1, l_0, i}&\cdots & B_{1, l_{n-d-2}, i}\\
\vdots&\vdots&\ddots&\vdots\\
S_{i} A_{r-1,i} & B_{r-1, l_0, i}&\cdots & B_{r-1, l_{n-d-2}, i}
\end{pmatrix}
\\&
\times\begin{pmatrix}
\mathbf{f}_i \\
R_{i,l_0}\mathbf{f}_{i,l_0}\\
\vdots\\
R_{i,l_{n-d-2}}\mathbf{f}_{i,l_{n-d-2}}\\
\end{pmatrix}=-\sum_{j\in H_i}\underbrace{\left(\begin{array}{c}
S_{i} A_{0,j}\\
S_{i} A_{1,j}\\
\vdots\\
S_{i} A_{r-1,j}
\end{array}\right)\mathbf{f}_j}_{\mathrm{interference ~by~}\mathbf{f}_j},
\end{align}
where $L_i=\{l_0,\ldots, l_{n-d-2}\}$ and $S_{i}A_{t,l}=B_{t, l, i}R_{i,l}$. Note that \eqref{eqn_alternatively_repair} was also introduced in \cite{liu2023generic}.

In the subsequent sections, based on the structures of the parity-check matrices, we check R2 and establish that \eqref{eqn_alternatively_repair} uniquely determines $\mathbf{f}_i$ when analyzing the repair property of the new codes $\mathcal{C}_1$ and $\mathcal{C}_2$, and check R1-R3 for that of $\mathcal{C}_3$.

Clearly, the repair bandwidth of node $i$ can be represented by
\begin{equation*}
\gamma_i=\sum\limits_{j\in H_i}{\rm rank}(R_{i,j}).
\end{equation*}
Node $i$ is said to have the optimal repair bandwidth if $\gamma_i={d\over d-k+1}N$, which is achieved when $\textrm{rank}(R_{i,j})={N\over d-k+1}$ for all $j\in H_i$.
In a scenario where a failed node is repaired with the restriction that some $d_c$ helper nodes are compulsory choices while the remaining $d-d_c$ helper nodes can be chosen freely, we say that the code exhibits the \textit{$(d,d_c)$-repair property}, which was originally introduced by Guruswami \textit{et al.} in \cite{Guruswami2020}.
Further, if $\gamma_i\le (1+\epsilon){d\over d-k+1}N$ for a small constant $0<\epsilon<1$, we say that node $i$ possesses the \textit{$(1+\epsilon)$-optimal repair bandwidth} \cite{li2023pmds}, also known as near-optimal repair bandwidth in \cite{Rawat}.

In addition to the repair bandwidth, the amount of data accessed during the repair process is also a crucial metric. Similarly,
the data access amount when repairing node $i$ can be calculated as $\sum\limits_{j\in H_i}N_c(R_{i,j})$, where $N_c(A)$ represents the number of nonzero columns in matrix $A$.
A code is said to have the \textit{optimal access property} if the amount of data accessed during the repair process is also ${d\over d-k+1}N$ \cite{Barg1}, and \textit{$(1+\epsilon)$-optimal access property} if the amount does not exceed $(1+\epsilon)\frac{d}{d-k+1}N$ for a small constant $0<\epsilon<1$.

\subsection{Partition of Basis $\{e_0,\cdots,e_{N-1}\}$}\label{sec:partition}

Let $w$ and $m$ be two integers greater than or equal to $2$. We denote by $e_0,\cdots,e_{w^m-1}$ the standard basis of $\mathbf{F}_q^{w^m}$, i.e.,
\begin{equation*}
e_i=(0,\cdots,0,1,0,\cdots,0),\,\,i\in [0: w^m),
\end{equation*}
where the $i$-th entry is 1 and all other entries are 0. Throughout this paper,
let $N=w^m$ unless otherwise stated. Then for any $a,b\in [0: N)$, we have
\begin{eqnarray}\label{Eqn e_ae_bT}
e_a(e_b)^{\top} &=& \left\{
\begin{array}{ll}
1, & \textrm{if $a=b$},\\
0, & \textrm{otherwise},
\end{array}
\right.
\end{eqnarray}
where $\top$ denotes the transpose operator.
In the following, we revisit a family of partitions of $\{e_0,\cdots,e_{w^m-1}\}$ originally introduced in \cite{wang2012long} and later reformulated in \cite{eMSR_d_eq_n-1}, which play an important role in our proposed new construction.

Given an integer $a\in [0: w^m)$, we denote its $w$-ary expansion by $(a_0,\cdots,a_{m-1})$
with $a_0$ being the most significant digit, i.e., $a=\sum\limits_{j=0}^{m-1}w^{m-1-j}a_{j}$. For convenience, we also write $a=(a_0,\cdots,a_{m-1})$.
For $i\in [0: m)$ and $t\in [0: w)$, we define a subset of $\{e_0,\cdots,e_{w^m-1}\}$ as
\begin{equation}\label{Eqn_Vt}
V_{i,t}=\{e_a|a_i=t, 0\le a< w^m\},
\end{equation}
where $a_i$ is the $i$-th element in the $w$-ary expansion of $a$.
Obviously, $V_{i,0},V_{i,1},\cdots,V_{i,w-1}$ form a partition of the set $\{e_0,\cdots,e_{w^m-1}\}$ and $|V_{i,t}|=w^{m-1}$ for any $i\in [0 : m)$. Examples of the set partitions defined in \eqref{Eqn_Vt} can be found in \cite{wang2012long} and \cite{eMSR_d_eq_n-1}.

For convenience of notation, we also denote by $V_{i,t}$ the $w^{m-1}\times w^m$ matrix, whose rows are formed by vectors $e_i$ in their corresponding sets, and $i$ is sorted in ascending order. We further define
\begin{equation}\label{Eqn_Vt2}
V_{i,t}=V_{i-m,t}, \mbox{~for~}i\ge m, t\in [0: w).
\end{equation}

In addition, we denote by $\hat{e}_0,\cdots,\hat{e}_{w^{m-1}-1}$ the standard basis of $\mathbf{F}_q^{w^{m-1}}$ and define
\begin{equation*}
\hat{V}_{i,t}=\{\hat{e}_a|a_i=t, 0\le a< w^{m-1}\}, 
\end{equation*}
for $i\in [0: m-1)$, $t\in [0: w)$,
and
\begin{equation*}
    \hat{V}_{i,t}=\hat{V}_{i-m+1,t}\mbox{~for~}i\ge m-1, t\in [0: w).
\end{equation*}

Furthermore, for any $i_1, i_2\ge0$ and $i_1\not \equiv i_2\mbox{\ mod\ } (m-1)$, we define $\hat{V}_{i_1,i_2,t_1,t_2}=\hat{V}_{i_2,i_1,t_2,t_1}=\hat{V}_{i_1,t_1}\cap \hat{V}_{i_2,t_2}$,
where $0\le t_1,t_2<w$.

\subsection{Notations}\label{sec:notations}
In the following, we define some useful notations and equalities, some of which were also introduced in \cite{li2023MDS}.

\begin{itemize}
    \item  $a(i, u)$: For $a=(a_0,\cdots,a_{m-1})\in [0: N)$, $i\in [0: m)$ and $u\in [0: w)$, we define $a(i, u)$ as
\begin{equation}\label{Eqn_ait}
a(i,u)=(a_0,\cdots,a_{i-1},u,a_{i+1},\ldots,a_{m-1}),
\end{equation}
which means replacing the $i$-th digit by $u$.

\item $g_{i,u}(a)$: For $a=(a_0,a_1,\cdots,a_{m-2})\in [0:N/w)$ and $i\in [0: m)$, define
\begin{equation}\label{Eqn_g}
g_{i,u}(a)=(a_0,a_1,\cdots,a_{i-1},u,a_i,\cdots,a_{m-2}),
\end{equation}
which means inserting $u$ into the $i$-th digit of $(a_0,a_1,\cdots,a_{m-2})$ to form an integer in $[0:N)$.

\item $K[a,b]$: The $(a,b)$-th entry of the matrix $K$.

\item $K[a:]$: The $a$-th row of the matrix $K$.

\item $i\%m$: The smallest nonnegative residue of $i$ modulo $m$ for positive integers $i$ and $m$, i.e., $i\% m\equiv i \bmod m$.

\item $I_{N}$: The identity matrix of order $N$ (abbreviated as $I$ when the order is clear).

\item $\oplus$:  Addition modulo $w$.
\end{itemize}

For $i,j\in [0:m)$ and $u,v \in [0: w)$, by \eqref{Eqn_g}, we have that the $j$-th digit of $g_{i,u}(a)$ is given by
\begin{eqnarray}\label{Eqn ga b}
\left(g_{i,u}(a)\right)_j=\left\{
\begin{array}{ll}
a_j, &\mbox{~if~}j<i,\\
u,&\mbox{~if~}j=i,\\
a_{j-1},&\mbox{~if~}j>i.
\end{array}
\right.
\end{eqnarray}
By \eqref{Eqn_ait}, replacing the $j$-th digit of $g_{i,u}(a)$ with $v$ leads to
\begin{eqnarray}\label{Eqn ga b (j,v)}
(g_{i,u}(a))(j,v)=\left\{
\begin{array}{ll}
g_{i,u}(a(j,v)), &\mbox{~if~}j<i,\\
g_{i,v}(a),&\mbox{~if~}j=i,\\
g_{i,u}(a(j-1,v)),&\mbox{~if~}j>i.
\end{array}
\right.
\end{eqnarray}

\section{A New $(n, k)$ MDS Array Code With Repair Degree $d<n-1$}

\begin{Construction}\label{Con C1}
For any $m,s\ge 2$, $2<k<2ms-2$, let $n=2ms$, we construct an $(n,k)$ array code $\mathcal{C}_1$ over $\mathbf{F}_q$ with repair degree $d$ such that $\max\{s-1,k+1\}\le d<n-1$ and a sub-packetization level of $N=w^m$, where $w=d-k+1$. The parity-check matrix $(A_{t,i})_{t\in [0: r),i\in[0: n)}$ of the array code $\mathcal{C}_1$ is defined as
\begin{align}\label{Eqn_PC_New1}
\nonumber &\begin{pmatrix}
V_{i,0}\\
V_{i,1}\\
\vdots\\
V_{i,w-1}
\end{pmatrix}A_{t,i}\\=&\left\{\hspace{-2mm}\begin{array}{ll}
\begin{pmatrix}
\lambda_{i,0}^tV_{i,0}+\sum\limits_{u=1}^{w-1}(\lambda_{i,0}^t-\lambda_{i,u}^t)V_{i,u}\\
\lambda_{i,1}^tV_{i,1}\\
\vdots\\
\lambda_{i,w-1}^tV_{i,w-1}
\end{pmatrix}, & \textrm{if~}i\in [0:ms), \\
\begin{pmatrix}
\lambda_{i,0}^tV_{i,0}\\
\lambda_{i,1}^tV_{i,1}\\
\vdots\\
\lambda_{i,w-1}^tV_{i,w-1}
\end{pmatrix}, & \textrm{if~}i\in [ms:n), \\
\end{array}\right.
\end{align}
where $\lambda_{i,u}\in \mathbf{F}_q$ for $i\in [0: n)$ and $u\in [0: w)$.
The selection matrix and repair matrix of node $i$ are defined as
\begin{equation}\label{Eqn select matrix of code New1}
S_{i}=\left\{\begin{array}{ll}
V_{i,0}, & \textrm{if~}i\in [0:ms), \\
V_{i,0}+V_{i,1}+\cdots+V_{i,w-1}, & \textrm{if~}i\in [ms:n),
\end{array}\right.
\end{equation}
and
\begin{equation}\label{Eqn repair matrix of code New1}
R_{i,j}=\left\{\begin{array}{ll}
I_N, & \textrm{if~}j\equiv i\bmod m \textrm{~and~}\lfloor \frac{j}{ms}\rfloor=\lfloor \frac{i}{ms}\rfloor,\\
S_{i}, & \textrm{otherwise},
\end{array}\right.
\end{equation}
for distinct $i,j\in [0: n)$, where $V_{i,0},V_{i,1},\cdots,V_{i,w-1}$ are defined by \eqref{Eqn_Vt} and \eqref{Eqn_Vt2}.
\end{Construction}

Before presenting the proof, we first introduce a lemma from \cite{li2023MDS}, which is helpful when verifying the MDS property of the new array code $\mathcal{C}_1$.

\begin{Lemma}(\cite[Lemma 1]{li2023MDS})\label{lem the important for the proof of this section}
For $t, i\in [0:r)$, let $K_{t,i}$ be an $N\times N$ upper triangular matrix, i.e., 
\begin{equation}\label{Eqn_upper}
K_{t,i}[a,b]=0 \mbox{~for~} 0\le b<a<N.
\end{equation}
Then
the following block matrix
\begin{equation*}
K=\left(\begin{array}{cccc}
K_{0,0} & K_{0,1} & \cdots & K_{0,r-1}\\
K_{1,0} & K_{1,1} & \cdots & K_{1,r-1}\\
\vdots & \vdots & \vdots & \vdots\\
K_{r-1,0} & K_{r-1,1} & \cdots & K_{r-1,r-1}\\
\end{array}\right)
\end{equation*}
is non-singular
if
\begin{itemize}
\item [i)] $K_{t,i}[a,a]=(K_{1,i}[a,a])^t$ for $i, t\in [0:r)$ and $a\in [0: N)$,

\item [ii)] $K_{1,i}[a,a]\ne K_{1,j}[a,a]$ for any $i,j\in [0,r)$ with $j\ne i$ and $a\in [0: N)$.
\end{itemize}
\end{Lemma}

\begin{figure*}
  \begin{equation}\label{Eqn the first eq for proving Thm 9}
S_{i}A_{t,j}
=\left\{\begin{array}{ll}
\lambda_{j,0}^t V_{i,0}+\sum\limits_{u=1}^{w-1}(\lambda_{j,0}^t-\lambda_{j,u}^t )V_{i,u}, & \textrm{if~}i,j\in [0: ms) \textrm{~and~}i\equiv j \bmod m,\\
\lambda_{j,0}^t V_{i,0}+\lambda_{j,1}^t V_{i,1}+\cdots+\lambda_{j,w-1}^t V_{i,w-1}, & \textrm{if~}i,j\in [ms: n) \textrm{~and~}i\equiv j \bmod m,\\
B_{t,j,i}R_{i,j},&\textrm{if~}i\in [0:ms), j\in [ms:n)\textrm{~with~}i\equiv j \bmod m,\\
&\textrm{~or~}i\in [ms:n), j\in [0:ms) \textrm{~with~}i\equiv j \bmod m,\\
B_{t,j,i}R_{i,j}, & \textrm{if~}i\not\equiv j \bmod m,
\end{array}\right.
\end{equation}
\hrulefill
\end{figure*}

In the following, we analyze the repair property of code $\mathcal{C}_1$ by checking R2 and verifying that \eqref{eqn_alternatively_repair} is solvable using the downloaded data. Thus, it is helpful to characterize the product of $S_{i}$ and $A_{t,j}$ beforehand. We first define some preliminary matrices $B_{t,j,i}$. Let 
\begin{equation}\label{Eqn_com_set}
D_i=\{j\in [0:n)\setminus\{i\}|j\equiv i\bmod m \textrm{~and~}\lfloor \frac{j}{ms}\rfloor=\lfloor \frac{i}{ms}\rfloor\}.    
\end{equation}
For $i\in [0: n)$ and $j\in [0:n)\setminus D_i$, let $\bar{i},\bar{j}\in [0: m)$ such that $\bar{i}\equiv i \bmod m$ and $\bar{j}\equiv j \bmod m$, define the $N/w\times N/w$ matrix $B_{t,j,i}$ as follows:

If $\bar{j} <\bar{i}$, the matrix $B_{t,j,i}$ satisfies
\begin{align}\label{Eqn_B<}
   \nonumber & \begin{pmatrix}
\hat{V}_{\bar{j},0}\\
\hat{V}_{\bar{j},1}\\
\vdots\\
\hat{V}_{\bar{j},w-1}
\end{pmatrix}B_{t,j,i}\\
=&\left\{\hspace{-2mm}\begin{array}{ll}
\begin{pmatrix}
\lambda_{j,0}^t\hat{V}_{\bar{j},0}+\sum\limits_{u=1}^{w-1}(\lambda_{j,0}^t-\lambda_{j,u}^t)\hat{V}_{\bar{j},u}\\
\lambda_{j,1}^t\hat{V}_{\bar{j},1}\\
\vdots\\
\lambda_{j,w-1}^t\hat{V}_{\bar{j},w-1}
\end{pmatrix}, &\mbox{if~} j\in [0: ms),\\
\begin{pmatrix}
\lambda_{j,0}^t\hat{V}_{\bar{j},0}\\
\lambda_{j,1}^t\hat{V}_{\bar{j},1}\\
\vdots\\
\lambda_{j,w-1}^t\hat{V}_{\bar{j},w-1}
\end{pmatrix},&\mbox{if~} j\in [ms: n).
\end{array}\right.
\end{align}

If $\bar{j}>\bar{i}$,  the matrix $B_{t,j,i}$ satisfies
\begin{align}\label{Eqn_B>}
 \nonumber &   \begin{pmatrix}
\hat{V}_{\bar{j}-1,0}\\
\hat{V}_{\bar{j}-1,1}\\
\vdots\\
\hat{V}_{\bar{j}-1,w-1}
\end{pmatrix}B_{t,j,i}\\=&\left\{\hspace{-3mm}\begin{array}{ll}
 \begin{pmatrix}\hspace{-1.1mm}
\lambda_{j,0}^t\hat{V}_{\bar{j}-1,0}\hspace{-1mm}+\hspace{-1mm}\sum\limits_{u=1}^{w-1}(\lambda_{j,0}^t-\lambda_{j,u}^t)\hat{V}_{\bar{j}-1,u}\\
\lambda_{j,1}^t\hat{V}_{\bar{j}-1,1}\\
\vdots\\
\lambda_{j,w-1}^t\hat{V}_{\bar{j}-1,w-1}
\hspace{-2.5mm}\end{pmatrix}\hspace{-1mm}, &\hspace{-1mm}\mbox{if~} j\in [0: ms),\\
\begin{pmatrix}
\lambda_{j,0}^t\hat{V}_{\bar{j}-1,0}\\
\lambda_{j,1}^t\hat{V}_{\bar{j}-1,1}\\
\vdots\\
\lambda_{j,w-1}^t\hat{V}_{\bar{j}-1,w-1}
\end{pmatrix},&\hspace{-1mm}\mbox{if~} j\in [ms: n).
\end{array}\right.
\end{align} 

If $\bar{j}=\bar{i}$, and $i\in [0:ms)$, $j\in [ms:n)$ or $i\in [ms:n)$, $j\in [0:ms)$, the matrix $B_{t,j,i}$ is defined as
\begin{equation}\label{Eqn_B_equiv}
B_{t,j,i}=\lambda_{j,0}^tI_{N/w}.   
\end{equation}

By the definition of the matrix $B_{t,j,i}$, we have the following facts.
\begin{Fact}\label{Fact_B}
For any $i\in [0:n)$, $j\in [0: n)\setminus (D_i\cup\{i\})$, and $t\in [0: r)$, the $N/w\times N/w$ matrix $B_{t,j,i}$
 defined by \eqref{Eqn_B<}, \eqref{Eqn_B>}
and \eqref{Eqn_B_equiv} satisfies:
\begin{itemize}
\item [i)] When $b<a$, $B_{t,j,i}[a,b]=0$, i.e., $B_{t,j,i}$ is upper triangular;
\item [ii)] $B_{t,j,i}[a,a]=(B_{1,j,i}[a,a])^t$.
\end{itemize}    
\end{Fact}

\begin{Proposition}\label{prop:SA_C1}
For $i,j\in [0: n)$ and $t\in [0: r)$, we have \eqref{Eqn the first eq for proving Thm 9},
where the $N/w\times N/w$ matrix $B_{t,j,i}$ is defined by \eqref{Eqn_B<}, \eqref{Eqn_B>}, and \eqref{Eqn_B_equiv}.
\end{Proposition}
\begin{proof}
The proof is given in Appendix \ref{sec:pfC3}.
\end{proof}

\begin{figure*}
\begin{eqnarray}\label{Eqn proving R3 of C3_2nd_par}
\underbrace{\setlength{\arraycolsep}{1pt}\begin{pmatrix}
I_{N/w} &\textbf{0}_{N/w} & \cdots &\textbf{0}_{N/w}   & B_{0,l_0,i} & \cdots & B_{0,l_{r-w-1},i}\\
\lambda_{i,0}I_{N/w} & (\lambda_{i,0}-\lambda_{i,1})I_{N/w} & \cdots & (\lambda_{i,0}-\lambda_{i,w-1})I_{N/w} & B_{1,l_0,i} & \cdots & B_{1,l_{r-w-1},i}\\
\vdots & \vdots &\ddots & \vdots & \vdots & \ddots & \vdots \\
\lambda_{i,0}^{r-1}I_{N/w} &  (\lambda_{i,0}^{r-1}-\lambda_{i,1}^{r-1})I_{N/w} & \cdots & (\lambda_{i,0}^{r-1}-\lambda_{i,w-1}^{r-1})I_{N/w} & B_{r-1,l_0,i} & \cdots & B_{r-1,l_{r-w-1},i}\\
\end{pmatrix}}_{\mathrm{block~matrix~B''}}\begin{pmatrix}
V_{i,0}\mathbf{f}_i\\
\vdots\\
V_{i,w-1}\mathbf{f}_i\\
R_{i,l_0}\mathbf{f}_{l_0}\\
\vdots\\
R_{i,l_{r-w-1}}\mathbf{f}_{l_{r-w-1}}\\
\end{pmatrix}=\kappa^*,
\end{eqnarray}
    \hrulefill
\end{figure*}

\begin{Theorem}\label{Thm_C3}
The code $\mathcal{C}_1$ in Construction \ref{Con C1} is an $(n=2ms, k)$ MDS array code over $\mathbf{F}_q$ with the $(d=k+w-1, d_c=s-1)$ repair property and $(1+\frac{d_c(d-k)}{d})$-optimal repair bandwidth
if the following conditions can be satisfied.
\begin{itemize}
\item [i)] $\lambda_{i,u}\ne \lambda_{j,u'}$ for $u, u'\in [0: w)$ and $i,j\in [0: n)$ with $i\not\equiv j\bmod m$,
\item [ii)] $\lambda_{i,u}\ne \lambda_{j,u}$ for $u\in [0: w)$ and $i,j\in [0: n)$ with $i\ne j$ and $i\equiv j\bmod m$,
\item [iii)] $\lambda_{i,u}\ne \lambda_{i,u'}$ for $u,u' \in [0:w)$ with $u\ne u'$ and $i\in [0: n)$,
\item [iv)] $\lambda_{i,0}\ne \lambda_{j,u}$ for $u\in [1: w)$ and $i,j\in [0: n)$ with $i\equiv j\bmod m$ and $\lfloor \frac{i}{ms}\rfloor \ne \lfloor \frac{j}{ms}\rfloor$.
\end{itemize}
\end{Theorem}

\begin{proof}
For $i,j\in [0:n)$, we rewrite them as
$i=v_0m+\bar{i}$ and $j=v_1m+\bar{j}$ for $v_0,v_1\in [0: 2s)$ and $\bar{i},\bar{j}\in [0: m)$.

\textbf{MDS property:}
By \eqref{Eqn_PC_New1}, we have $A_{t,i}[a,b]=0$ for $0\le b<a$ and $A_{t,i}[a,a]=(A_{1,i}[a,a])^t$.

In addition, for $i,j \in [0: n)$ with $j\ne i$, and for $a\in [0: N)$, we have
\begin{equation*}
A_{1,i}[a,a]-A_{1,j}[a,a]=\lambda_{i,a_{\bar{i}}}-\lambda_{j,a_{\bar{j}}} \ne 0,
\end{equation*}
where the inequality follows from i) if $i\not\equiv j\bmod m$ (i.e., $\bar{i}\ne \bar{j}$) and ii) if $i\equiv j\bmod m$ (i.e., $\bar{i}= \bar{j}$).
By applying Lemma \ref{lem the important for the proof of this section}, we can derive that the new code $\mathcal{C}_1$ has the MDS property if conditions i) and ii) are satisfied.

\textbf{Repair property:}  Let $H_i$ be any $d$-subset of $[0:n)\setminus\{i\}$ containing $D_i$ in \eqref{Eqn_com_set}. Suppose that node $i$ fails where $i\in [ms: n)$, following the repair mechanism in Section \ref{sec:repair}, 
by applying Proposition \ref{prop:SA_C1} to  \eqref{eqn_alternatively_repair},
we obtain
\begin{align}\label{Eqn proving R3 of C3}
\nonumber&\setlength{\arraycolsep}{1pt}\underbrace{\left(\begin{array}{cccccc}\hspace{-1mm}
I_{N/w} & \cdots & I_{N/w} & B_{0,l_0,i} & \cdots & B_{0,l_{r-w-1},i}\\
\lambda_{i,0}I_{N/w} & \cdots & \lambda_{i,w-1}I_{N/w} & B_{1,l_0,i} & \cdots & B_{1,l_{r-w-1},i}\\
\vdots & \ddots & \vdots & \vdots & \ddots & \vdots \\
\lambda_{i,0}^{r-1}I_{N/w} & \cdots & \lambda_{i,w-1}^{r-1}I_{N/w} & B_{r-1,l_0,i} & \cdots & B_{r-1,l_{r-w-1},i}\\
\hspace{-1mm}\end{array}\right)}_{\mathrm{block~matrix~B'}}\\
&\times\left(\begin{array}{c}
V_{i,0}\mathbf{f}_i\\
\vdots\\
V_{i,w-1}\mathbf{f}_i\\
R_{i,l_0}\mathbf{f}_{l_0}\\
\vdots\\
R_{i,l_{r-w-1}}\mathbf{f}_{l_{r-w-1}}\\
\end{array}\right)=\kappa^*,
\end{align}
where we note that $n-d-1=r-w$ and
\begin{align*}
\kappa^*=& -\hspace{-1mm}\sum\limits_{\substack {j\in [ms: n)\setminus\{i\}\\ j\equiv i \bmod m}}
\setlength{\arraycolsep}{1pt}\left(\hspace{-1mm}\begin{array}{c}
V_{i,0}+\cdots+V_{i,w-1}\\
\lambda_{j,0} V_{i,0}+\cdots+\lambda_{j,w-1} V_{i,w-1}\\
\vdots\\
\lambda_{j,0}^{r-1} V_{i,0}+\cdots+\lambda_{j,w-1}^{r-1} V_{i,w-1}\\
\end{array}\hspace{-1mm}\right)\mathbf{f}_j\\
&-\sum\limits_{\substack{j\in H_i\cap [ms: n), j\not\equiv i \bmod m\\ \mbox{or~} j\in H_i\cap [0:ms)}}
\left(\begin{array}{c}
B_{0,j,i}\\
B_{1,j,i}\\
\vdots\\
B_{r-1,j,i}\\
\end{array}\right)R_{i,j}\mathbf{f}_j,
\end{align*}
which is a column vector of length $rN/w$ and can be determined from the downloaded data.

By Lemma \ref{lem the important for the proof of this section} and Fact \ref{Fact_B}, we have that
the block matrix $B'$ in \eqref{Eqn proving R3 of C3} is non-singular if
\begin{equation*}
\lambda_{i,0},\hspace{-0.5mm} \ldots, \hspace{-0.5mm}\lambda_{i,w-1}, \hspace{-0.5mm}B_{1,l_0,i}[a,a], \hspace{-0.5mm}B_{1,l_1,i}[a,a], \hspace{-0.5mm}\ldots, \hspace{-0.5mm}B_{1,l_{r-w-1},i}[a, a]
\end{equation*}
are pairwise distinct for any $a\in [0: N/w)$ and $$L_i\subset [0: n)\backslash \{l|l\in [ms:n), l\equiv i\bmod m\}.$$ It suffices to require that
\begin{equation*}
\lambda_{i,0}, \lambda_{i,1}, \ldots, \lambda_{i,w-1},\end{equation*}
and \begin{equation*}
B_{1,j,i}[a, a], j\in [0: n)\backslash \{l|l\in [ms:n), l\equiv i\bmod m\}
\end{equation*}
are pairwise distinct for any  $a\in [0: N/w)$ since $L_i$ is an arbitrary $(r-w)$-subset of $[0: n)\backslash \{l|l\in [ms:n), l\equiv i\bmod m\}$, which can be fulfilled if 
\begin{itemize}
    \item 
[1)] 
$\lambda_{i,0}, \lambda_{i,1}, \ldots, \lambda_{i,w-1}$
are pairwise distinct,
\item [2)] 
$\lambda_{i,u}\ne \lambda_{j,v}$ for $j\in [0:n)$ with $j\not\equiv i\bmod m$ and $u,v\in [0:w)$
according to \eqref{Eqn_B<} and \eqref{Eqn_B>}
and \begin{equation*}
\lambda_{i,u}\ne \lambda_{j,0} \mbox{~for~} j\in [0:ms),   u\in [0:w), j\equiv i\bmod m
\end{equation*}
by \eqref{Eqn_B_equiv},

\item [3)]  
a. $\lambda_{j,u}\ne \lambda_{j',v}$ for $u,v\in [0:w)$ and 
\begin{equation*}
j,j'\in [0:n)\backslash \{l|l\in [ms:n), l\equiv i\bmod m\}
\end{equation*}
with
$j,j'\not\equiv i\bmod m\mbox{~and~} j\not\equiv j'\bmod m $
according to \eqref{Eqn_B<} and \eqref{Eqn_B>}.\\
b. $\lambda_{j,u}\ne \lambda_{j',u}$ for $u\in [0:w)$ and $$j,j'\in [0:n)\backslash \{l|l\in [ms:n), l\equiv i\bmod m\}$$ with $j,j'\not\equiv i\bmod m$, $j\equiv j'\bmod m$, and $j\ne j'$.\\
c. $\lambda_{j,u}\ne \lambda_{j',0}$ for $u\in [0:w)$ and
\begin{equation*}
j,j'\in [0:n)\backslash \{l|l\in [ms:n), l\equiv i\bmod m\}     
\end{equation*} with $j\not\equiv i\bmod m$ and $j'\equiv i\bmod m$,
and\\
d. $\lambda_{j,0}\ne \lambda_{j',0}$ for 
\begin{equation*}
j,j'\in [0:n)\backslash \{l|l\in [ms:n), l\equiv i\bmod m\} 
\end{equation*}
with $j, j'\equiv i\bmod m$ and $j\ne j'$ by \eqref{Eqn_B_equiv}.
\end{itemize}

Requirement 1) can be satisfied if iii) holds, Requirement 2) ensures that $\lambda_{i,u}\ne  B_{1,j,i}[a, a]$ for any $u\in [0: w)$, $a\in [0:N/w)$ and $j\in [0: n)\backslash \{l|l\in [ms:n), l\equiv i\bmod m\}$, and can be satisfied if i), ii), and iv) hold.
Requirement 3) ensures that
$B_{1,j,i}[a, a]\ne B_{1,j',i}[a, a]$ for any $a\in [0:N/w)$ and different $j,j'\in [0: n)\backslash \{l|l\in [ms:n), l\equiv i\bmod m\}$, which holds if i), ii), and iv) hold.

Therefore, if i)--iv) hold, then the block matrix $B'$ in \eqref{Eqn proving R3 of C3} is non-singular, thus we can solve $V_{i,0}\mathbf{f}_i,\cdots,V_{i,w-1}\mathbf{f}_i$ (i.e., $\mathbf{f}_i$) from \eqref{Eqn proving R3 of C3}.

Suppose that node $i$ fails where $i\in [0: ms)$,
by applying Proposition \ref{prop:SA_C1} to  \eqref{eqn_alternatively_repair},
we obtain \eqref{Eqn proving R3 of C3_2nd_par},
where 
\begin{align*}
\kappa^*&= -\sum\limits_{\substack{j\in [0: ms)\setminus\{i\}\\ j\equiv i \bmod m}}
\left(\hspace{-2mm}\begin{array}{c}
V_{i,0}\\
\lambda_{j,0} V_{i,0}+\sum\limits_{u=1}^{w-1}(\lambda_{j,0}-\lambda_{j,u})V_{i,u}\\
\vdots\\
\lambda_{j,0}^{r-1} V_{i,0}+\sum\limits_{u=1}^{w-1}(\lambda_{j,0}^{r-1}-\lambda_{j,u}^{r-1})V_{i,u}\\
\end{array}\hspace{-2mm}\right)\mathbf{f}_j\\
&-\sum\limits_{\substack{j\in H_i\cap [0:ms), j\not\equiv i \bmod m\\ \mbox{or~} j\in H_i\cap [ms:n)}}
\left(\begin{array}{c}
B_{0,j,i}\\
B_{1,j,i}\\
\vdots\\
B_{r-1,j,i}\\
\end{array}\right)R_{i,j}\mathbf{f}_j,
\end{align*}
which is a column vector of length $rN/w$ and can be determined from the downloaded data.

Clearly, under elementary column operations, the block matrix $B''$ in \eqref{Eqn proving R3 of C3_2nd_par} is equivalent to the block matrix $B'$ in \eqref{Eqn proving R3 of C3}.

Similarly, the block matrix $B''$ is non-singular if $\lambda_{i,0}, \lambda_{i,1}, \ldots, \lambda_{i,w-1}$, and
\begin{equation*}
B_{1,j,i}[a, a], j\in [0: n)\backslash \{l|l\in [0:ms), l\equiv i\bmod m\}
\end{equation*}
are pairwise distinct for any  $a\in [0: N/w)$, 
which can be fulfilled if i)--iv) hold.

This completes the proof.
\end{proof}

In the following, we give a concrete example of the code in Construction \ref{Con C1} and the node repair process.
\begin{Example}
Let $w=m=s=2$, $k=5$, and $d=6$, then the block entries of the parity-check matrix $(A_{t,i})_{t\in[0: 3), i\in [0: 8)}$ of the $(n=8, k=5, d=6)$ array code $\mathcal{C}_1$ in Construction \ref{Con C1} are given by 
\begin{equation*}
A_{t,i}=\begin{pmatrix}
\lambda_{i,0}^t   & & \lambda_{i,0}^t-\lambda_{i,1}^t & \\
&\lambda_{i,0}^t    && \lambda_{i,0}^t-\lambda_{i,1}^t  \\
&   &\lambda_{i,1}^t   & \\
&&   &\lambda_{i,1}^t    \\
\end{pmatrix} \mbox{~for~} i=0, 2,
\end{equation*}
\begin{equation*}
A_{t,i}=\begin{pmatrix}
\lambda_{i,0}^t   &   \lambda_{i,0}^t-\lambda_{i,1}^t && \\
&\lambda_{i,1}^t   & &   \\
&&\lambda_{i,0}^t   & \lambda_{i,0}^t-\lambda_{i,1}^t  \\
&&&\lambda_{i,1}^t       \\
\end{pmatrix} \mbox{~for~}   i=1, 3,    
\end{equation*}
\begin{equation*}
A_{t,i}=\begin{pmatrix}
\lambda_{i,0}^t   & &   & \\
&\lambda_{i,0}^t    &&   \\
&   &\lambda_{i,1}^t   & \\
&&   &\lambda_{i,1}^t    \\
\end{pmatrix}~   \mbox{for~}   i=4, 6, 
\mbox{and~} 
\end{equation*}
\begin{equation*}
A_{t,i}=\begin{pmatrix}
\lambda_{i,0}^t   &     && \\
&\lambda_{i,1}^t   & &   &\\
&&\lambda_{i,0}^t   &    \\
&&&\lambda_{i,1}^t       \\
\end{pmatrix} ~\mbox{for~}   i=5, 7.       
\end{equation*}

For illustration, suppose node $4$ fails, then node $6$ is a compulsory helper node. The selection matrix and repair matrix of node $4$ are
\begin{equation*}
S_{4}=\begin{pmatrix}
1&0&0&0\\
0&1&0&0
\end{pmatrix} +\begin{pmatrix}
0&0&1&0\\
0&0&0&1
\end{pmatrix},   
\end{equation*}
and
\begin{equation*}
R_{4,j}=\left\{\begin{array}{cl}
I_4,    & \mbox{if}~ j=6, \\[3pt]
S_4,      & \mbox{if}~ j\in[0: 8)\backslash\{4, 6\}.
\end{array}\right.  
\end{equation*}
The matrices $B_{t,j,4}$ in \eqref{Eqn_B>} and \eqref{Eqn_B_equiv} are
\begin{equation*}
B_{t,j,4}=\left\{\begin{array}{cc}
\begin{pmatrix}
  \lambda_{j,0}^t&\lambda_{j,0}^t-\lambda_{j,1}^t\\
  &\lambda_{j,1}^t
\end{pmatrix},     & \mbox{if}~ j=1,3, \\[3pt]
\begin{pmatrix}
  \lambda_{j,0}^t&\\
  &\lambda_{j,1}^t
\end{pmatrix},       & \mbox{if}~ j=5,7,
\end{array}\right.  
\end{equation*}
and
\begin{equation*}
    B_{t,j,4}=\lambda_{j,0}^t I_2, \mbox{~for~}j=0,2.
\end{equation*}

Next, we calculate the product of $S_i$ and $A_{t,j}$. 
For $j=0, 2$, we have
\begin{align}\label{Eqn_S4A0}
S_{4}A_{t,j}&= \begin{pmatrix}
    \lambda_{j,0}^t   & &   \lambda_{j,0}^t & \\
&\lambda_{j,0}^t   & &   \lambda_{j,0}^t   
\end{pmatrix} =\underbrace{\lambda_{j,0}^tI_2}_{B_{t,j,4}}R_{4,j}.
\end{align}

For $j=1, 3$, we have
\begin{align}\label{Eqn_S4A1}
\nonumber S_{4}A_{t,j}&= \begin{pmatrix}
    \lambda_{j,0}^t   &   \lambda_{j,0}^t-\lambda_{j,1}^t &\lambda_{j,0}^t   & \lambda_{j,0}^t-\lambda_{j,1}^t \\
&\lambda_{j,1}^t   & &   \lambda_{j,1}^t    
\end{pmatrix} \\
\nonumber &=\begin{pmatrix}
  \lambda_{j,0}^t&\lambda_{j,0}^t-\lambda_{j,1}^t\\
  &\lambda_{j,1}^t
\end{pmatrix} \begin{pmatrix}
1&0&1&0\\
0&1&0&1
\end{pmatrix}\\
&=B_{t,j,4}R_{4,j}.
\end{align}

For $j=5,7$, we have
\begin{align}\label{Eqn_S4A5}
S_{4}A_{t,j}&= \begin{pmatrix}
    \lambda_{j,0}^t   &  &  \lambda_{j,0}^t&  \\
&\lambda_{j,1}^t   &   & \lambda_{j,1}^t   
\end{pmatrix} =\underbrace{\begin{pmatrix}
    \lambda_{j,0}^t   &     \\
&\lambda_{j,1}^t    
\end{pmatrix}}_{B_{t,j,4}}R_{4,j}.
\end{align}

For $j=4, 6$, we have
\begin{align}\label{Eqn_S4A6}
S_{4}A_{t,j}&= \begin{pmatrix}
    \lambda_{j,0}^t   &  &  \lambda_{j,1}^t&  \\
&\lambda_{j,0}^t   &   & \lambda_{j,1}^t   
\end{pmatrix} =\lambda_{j,0}^tV_{4,0} +\lambda_{j,1}^t V_{4,1}.
\end{align} 
Clearly, \eqref{Eqn_S4A0}--\eqref{Eqn_S4A6} are consistent with \eqref{Eqn the first eq for proving Thm 9}.

For illustration, suppose that node 1 is not contacted during the repair process of node $4$, by \eqref{Eqn repairment gene}, we have

\begin{align*}
&\underbrace{\left(\begin{array}{c}
S_{4} A_{0,4} \\
S_{4} A_{1,4}\\
S_{4} A_{2,4}
\end{array}\right)\mathbf{f}_4}_{\mathrm{useful ~data}}+\underbrace{\left(\begin{array}{c}
S_{4} A_{0,1}\\
S_{4} A_{1,1}\\
S_{4} A_{2,1}
\end{array}\right)\mathbf{f}_1}_{\mathrm{unknown ~data~by~}\mathbf{f}_1}\\
&\hspace{3cm}
+\sum_{j\in [0: 8)\setminus\{1,4\}}\underbrace{\left(\begin{array}{c}
S_{4} A_{0,j}\\
S_{4} A_{1,j}\\
S_{4} A_{2,j}
\end{array}\right)\mathbf{f}_j}_{\mathrm{interference ~by~}\mathbf{f}_j}=0.
\end{align*}
Applying \eqref{Eqn_S4A0}--\eqref{Eqn_S4A6} to the above equation, we obtain
\begin{align}\label{Eqn_ex_R3}
\nonumber&\underbrace{\left(\begin{array}{cccccc}
\begin{psmallmatrix}
    1&\\
    &1
\end{psmallmatrix} & \begin{psmallmatrix}
    1&\\
    &1
\end{psmallmatrix} & \begin{psmallmatrix}
  1& \\
  &1
\end{psmallmatrix} \\
\begin{psmallmatrix}
    \lambda_{4,0}&\\
    &\lambda_{4,0}
\end{psmallmatrix} &   \begin{psmallmatrix}
    \lambda_{4,1}&\\
    &\lambda_{4,1}
\end{psmallmatrix} & \begin{pmatrix}
  \lambda_{1,0}&\lambda_{1,0}-\lambda_{1,1}\\
  &\lambda_{1,1}
\end{pmatrix}\\
\begin{psmallmatrix}
    \lambda_{4,0}^{2}&\\
    &\lambda_{4,0}^{2}
\end{psmallmatrix} &   \begin{psmallmatrix}
    \lambda_{4,1}^{2}&\\
    &\lambda_{4,1}^{2}
\end{psmallmatrix} & \begin{pmatrix}
  \lambda_{1,0}^{2}&\lambda_{1,0}^{2}-\lambda_{1,1}^{2}\\
  &\lambda_{1,1}^{2}
\end{pmatrix}\\
\end{array}\right)}_{\mathrm{block~matrix~B'}}\\
&\hspace{4cm}\times \left(\begin{array}{c}
V_{4,0}\mathbf{f}_4\\
V_{4,1}\mathbf{f}_4\\
R_{4,1}\mathbf{f}_{1}
\end{array}\right)=\kappa^*,
\end{align}
where 
\begin{align*}
\kappa^*=& -
\left(\begin{array}{c}
V_{4,0}+V_{4,1}\\
\lambda_{6,0} V_{4,0}+\lambda_{6,1} V_{4,1}\\
\lambda_{6,0}^2 V_{4,0}+\lambda_{6,1}^2 V_{4,1}
\end{array}\right)\mathbf{f}_6\\
&-\sum\limits_{j\in\{0,2,3,5,7\}}
\left(\begin{array}{c}
B_{0,j,4}\\
B_{1,j,4}\\
B_{2,j,4}\\
\end{array}\right)R_{4,j}\mathbf{f}_j,
\end{align*}
which is a column vector of length $6$ and can be determined from the downloaded data $R_{4,j}\mathbf{f}_j$, $j\in\{0,2,3,5,7\}$, and $\mathbf{f}_6$.

By Lemma \ref{lem the important for the proof of this section} and Fact \ref{Fact_B}, we have that
the block matrix $B'$ in \eqref{Eqn_ex_R3} is non-singular if
\begin{equation*}
\lambda_{4,0}, \lambda_{4,1}, \lambda_{1,u}
\end{equation*}
are pairwise distinct for $u=0,1$.

Therefore, node $4$ has the $(d=6, d_c=1)$ repair property and $(1+\frac{1}{6})$-optimal repair bandwidth. The repair for other nodes can be verified similarly.
\end{Example}

\begin{Theorem}\label{Thm_C3_field}

 
For $i\in [0: sm)$, rewrite $i=i_0\lceil s/w\rceil m+i_1$, where $i_0\in [0:w)$ and $i_1\in [0:\lceil s/w\rceil m)$, then the requirements in items i) - iv) of Theorem \ref{Thm_C3} can be fulfilled by setting $\lambda_{i,u}$ as
\begin{equation}\label{Eqn_C1_lambda}
\lambda_{i,u}=\left\{\begin{array}{ll}
c^{wi_1+(u+i_0)\%w},& {\rm if~} i\in [0:ms),\\
c^{\lceil s/w\rceil mw}\lambda_{i-ms,u},& {\rm if~} i\in [ms:n).\\
\end{array}\right.
\end{equation}
 where $c$ is a primitive element of
$\mathbf{F}_q$ with $q>2mw\lceil s/w\rceil$ and $\%$ is the modulo operation defined in Section \ref{sec:notations}.
\end{Theorem}

\begin{proof}
For $i,j\in [0: ms)$ with $i\ne j$, rewrite them as $i=i_0\lceil s/w\rceil m+i_1$ and $j=j_0\lceil s/w\rceil m+j_1$, where $i_0,j_0\in [0: w)$ and $i_1, j_1\in [0: \lceil s/w\rceil m)$, then we verify the four requirements in Theorem \ref{Thm_C3} as follows:
\begin{itemize}
 \item [i)]  If $i\not\equiv j\bmod m$, then $i_1\ne j_1$.
For $u,u'\in [0: w)$, we have

\begin{align*}
&\lambda_{i,u}-\lambda_{j,u'}\\
=& c^{wi_1+(u+i_0)\%w} -c^{wj_1+(u'+j_0)\%w}\\=&c^{wj_1+(u'+j_0)\%w}\left(c^{w(i_1-j_1)+(u+i_0)\%w-(u'+j_0)\%w}-1\right)\\\ne&0
\end{align*}
since
\begin{align*}
0&<\left|w(i_1-j_1)+(u+i_0)\%w-(u'+j_0)\%w\right|\\
&\le w\left(\lceil s/w\rceil m-1\right)+w-1\\
&=\lceil s/w\rceil wm-1\\
&<q-1.    
\end{align*}

Similarly, we also have $\lambda_{i+ms,u}\ne \lambda_{j+ms,u'}$.

By \eqref{Eqn_C1_lambda}, we have
\begin{equation}\label{Eqn_range_lamb1}
\lambda_{i,u}\in \left\{c^l|0
\le l<\lceil s/w\rceil m w\right\}
\end{equation}
and
\begin{equation}\label{Eqn_range_lamb2}
\lambda_{j+ms,u'}\in \left\{c^l|\lceil s/w\rceil m w
\le l<2\lceil s/w\rceil m w \right\},
\end{equation}
which indicates $\lambda_{i,u}\ne \lambda_{j+ms,u'}$ since $q>2mw\lceil s/w\rceil$.

Combining the above analysis, Theorem \ref{Thm_C3}-i) is satisfied.

\item [ii)] Consider $i \equiv j\bmod  m$, for $u\in [0: w)$, if $i_1=j_1$, then $i_0\ne j_0$,  we have
\begin{align*}
&\lambda_{i,u}-\lambda_{j,u}\\=&c^{wi_1+(u+i_0)\%w} -c^{wj_1+(u+j_0)\%w}\\
=&c^{wj_1+(u+j_0)\%w}\left(c^{(u+i_0)\%w-(u+j_0)\%w}-1\right)\\
\ne&0
\end{align*}
since $i_0\ne j_0$. 
If $i_1\ne j_1$, we have
\begin{align*}
&\lambda_{i,u}-\lambda_{j,u}\\
=&c^{wi_1+(u+i_0)\%w} -c^{wj_1+(u+j_0)\%w}\\
=&c^{wj_1+(u+j_0)\%w}\left(c^{w(i_1-j_1)+(u+i_0)\%w-(u+j_0)\%w}-1\right)\\
\ne&0.
\end{align*}

Similarly, we have $\lambda_{i+ms,u}\ne \lambda_{j+ms,u}$ and $\lambda_{i,u}\ne \lambda_{j+ms,u}$.
That is, Theorem \ref{Thm_C3}-ii) is satisfied.

\item [iii)] By \eqref{Eqn_C1_lambda}, it is obvious that $\lambda_{i,u}\ne \lambda_{i,u'}$ for $u,u' \in [0:w)$ with $u\ne u'$ and $i\in [0: n)$, i.e., Theorem \ref{Thm_C3}-iii) is satisfied
\item [iv)] If $i\equiv j\bmod m$, for $u\in [1: w)$, we have 
$\lambda_{i,0}\ne \lambda_{j+ms,u}$ and $\lambda_{i+ms,0}\ne \lambda_{j,u}$
according to \eqref{Eqn_range_lamb1} and \eqref{Eqn_range_lamb2}.
That is, Theorem \ref{Thm_C3}-iv) is satisfied.
\end{itemize}
\end{proof}
In the following, we present a concrete example of the values of $\lambda_{i,u}$ for the code $\mathcal{C}_1$ 
\begin{Example}
For the code $\mathcal{C}_1$ with $n=28$, $w=3$, $m=2$, and $s=7$, let $q=37$, which is larger than $2mw\lceil s/w\rceil$, and let $c$ be a primitive element of $\mathbf{F}_{37}$, then the values of $\lambda_{i,u}$, $i\in [0:28)$ and $u\in [0: 3)$  can be set as in Table \ref{Table ex_lamb} according to \eqref{Eqn_C1_lambda}.
\begin{table*}[htbp]
\begin{center}
\caption{The values of $\lambda_{i,u}$ for code $\mathcal{C}_1$ with $s=7$, $m=2$, $w=3$ and $n=2ms=28$.}\label{Table ex_lamb}
\renewcommand{\arraystretch}{1.5}
\begin{tabular}{|c|c|c|c|c|c|c||c|c|c|c|c|c||c|c|}
\hline
$i$ &$0$ & $1$& $2$ & $3$ & $4$ & $5$& $6$ & $7$ & $8$ & $9$& $10$ & $11$& $12$ & $13$ \\
\hline
$\lambda_{i,0}$ &$1$ &$c^3$ & $c^6$ &$c^9$& $c^{12}$& $c^{15}$ & $c$& $c^4$& $c^7$& $c^{10}$& $c^{13}$& $c^{16}$& $c^{2}$& $c^{5}$\\
\hline
$\lambda_{i,1}$ &$c$ &$c^{4}$ & $c^7$ & $c^{10}$ & $c^{13}$& $c^{16}$&  $c^2$&$c^5$&$c^8$&$c^{11}$&$c^{14}$&$c^{17}$&$1$&$c^{3}$\\\hline
$\lambda_{i,2}$ &$c^2$ &$c^{5}$ & $c^8$ & $c^{11}$ & $c^{14}$& $c^{17}$&  $1$&$c^3$&$c^6$&$c^9$&$c^{12}$&$c^{15}$&$c$&$c^{4}$\\
\hline
$i$ &$14$ & $15$& $16$ & $17$ & $18$ & $19$& $20$ & $21$ & $22$ & $23$& $24$ & $25$& $26$ & $27$ \\
\hline
$\lambda_{i,0}$ &$c^{18}$ &$c^{21}$ & $c^{24}$ &$c^{27}$& $c^{30}$& $c^{33}$ & $c^{19}$& $c^{22}$& $c^{25}$& $c^{28}$& $c^{31}$& $c^{34}$& $c^{20}$& $c^{23}$\\
\hline
$\lambda_{i,1}$ &$c^{19}$ &$c^{22}$ & $c^{25}$ & $c^{28}$ & $c^{31}$& $c^{34}$&  $c^{20}$&$c^{23}$&$c^{26}$&$c^{29}$&$c^{32}$&$c^{35}$&$c^{18}$&$c^{21}$\\\hline
$\lambda_{i,2}$ &$c^{20}$ &$c^{23}$ & $c^{26}$ & $c^{29}$ & $c^{32}$& $c^{35}$&  $c^{18}$&$c^{21}$&$c^{24}$&$c^{27}$&$c^{30}$&$c^{33}$&$c^{19}$&$c^{22}$\\
\hline
\end{tabular}
\end{center}
\end{table*}
\end{Example}

\section{Two New $(n, k)$ MDS Array Codes With Repair Degree $d<n-1$ Based on a Generic Transformation}\label{sec:transformation}

\subsection{Revisiting the Generic Transformation in \cite{eMSR_d_eq_n-1}}\label{sec:generic tran}

In this subsection, we first revisit the generic transformation introduced in \cite{eMSR_d_eq_n-1}.

\textbf{Generic transformation \cite{eMSR_d_eq_n-1}:} Let $(\overline{A}_{t,i})_{t\in[0: r), i\in[0: \overline{n})}$ be the parity-check matrix of an $(\overline{n},\overline{k})$ MSR code defined in the form of \eqref{Eqn parity check eq}, where $r=\overline{n}-\overline{k}$. The $\frac{N}{\overline{d}-\overline{k}+1}\times N$ matrices $\overline{R}_{i,j}$ and $\overline{S}_{i}$ denote the repair matrix and selection matrix, respectively, where $i,j\in[0: \overline{n})$ with $j\ne i$. A new $(n, k=n-r)$ array code over a certain finite field $\textbf{F}_{q}$ is obtained, where $n=s\overline{n}$ is assumed for convenience with $s\ge 2$.
The parity-check matrix $(A_{t,i})_{t\in [0:r),i\in[0:n)}$, the repair matrices $R_{i,j}$, and the selection matrices $S_{i}$ are defined
as follows:
\begin{eqnarray}
A_{t,j}&=&x_{t,j}\overline{A}_{t,j\%\overline{n}},\label{Eqn general coding matrix}\\
R_{i,j}&=&\left\{
\begin{array}{ll}
\overline{R}_{i\% \overline{n},j\% \overline{n}}, &\mbox{\ if\ } j\not\equiv i \bmod \overline{n}, \\
I_N, & \mbox{\ otherwise,\ }\\
\end{array}
\right.\label{Eqn general R}
\end{eqnarray}
and
\begin{eqnarray}\label{Eqn general S}
S_{i}=\overline{S}_{i\% \overline{n}},
\end{eqnarray}
where $x_{t,j}\in \mathbf{F}_q\backslash\{0\}$, $t\in[0:r)$,
$i,j\in[0:n)~\mbox{with}~j\ne i$, $\%$ is the modulo operation defined in Section \ref{sec:notations}.

By \eqref{Eqn general R}, we have $R_{i,j}=I$ for $j \equiv i \bmod \overline{n}$ and $j\ne i$. That is, when repairing node $i\in [0: n)$, it is necessary to contact the nodes with indices in the set $\{j|0\le j<n, j\ne i, j \equiv i \bmod \overline{n}\}$ and download all the data from these $d_c=s-1$ compulsory helper nodes. The remaining $d-d_c$ helper nodes can be chosen arbitrarily from the set $\{j|0\le j<n, j \not\equiv i \bmod \overline{n}\}$.

The generic transformation
in \cite{eMSR_d_eq_n-1} aimed to convert MSR codes with a repair degree of $\overline{d}=\bar{n}-1$ to MDS array codes with $d=n-1$. In this study, we discover that this transformation can also be applied to certain MSR codes with a repair degree of $\overline{d}<\bar{n}-1$ to yield MDS array codes with $d<n-1$.

%


We present the following result regarding the repair property of the new code obtained by applying the generic transformation in \cite{eMSR_d_eq_n-1} to an MSR code with a repair degree of $\bar{d}<\bar{n}-1$.

\begin{Theorem}\label{Thm general repair}
For the new $(n,k)$ array code over $\mathbf{F}_q$ obtained by applying the generic transformation in \cite{eMSR_d_eq_n-1} to an $(\bar{n}, \bar{k})$ MSR code with a repair degree of $\bar{d}<\bar{n}-1$, it exhibits the $(d=n-\overline{n}+\overline{d}, d_c=s-1)$-repair property and $(1+\frac{d_c(d-k)}{d})$-optimal repair bandwidth if R3 holds for the new $(n,k)$ code or \eqref{eqn_alternatively_repair} uniquely determines the failed node.
\end{Theorem}

\begin{proof}
Following a similar analysis to that in \cite[Theorem 2]{eMSR_d_eq_n-1}, it is observed that R1 and R2 can be preserved when applying the generic transformation in \cite{eMSR_d_eq_n-1} to MSR codes with a repair degree of $\overline{d}<\bar{n}-1$.
Thus, following the repair mechanism in Section \ref{sec:repair}, if R3 holds for the new resultant code or \eqref{eqn_alternatively_repair} is solvable, then node $i$ can be repaired with
a repair bandwidth of
\begin{align*}
\gamma_i &= \sum\limits_{j\in H_i}\mbox{rank}(R_{i,j})\\&=(d-d_c)\frac{N}{d-k+1}+d_cN\\&=(1+\frac{d_c(d-k)}{d})\gamma_{\rm optimal}\\
&<(1+r/\overline{n})\gamma_{\rm optimal},
\end{align*}
where $\gamma_{\rm optimal}=\frac{dN}{d-k+1}$ is the optimal value for the repair bandwidth of $(n, k)$ MDS array codes with repair degree $d$, and the less than sign holds since
\begin{eqnarray*}
\frac{d_c(d-k)}{d}=\frac{(s-1)(d-k)}{s\overline{n}-\overline{n}+\overline{d}}< \frac{(d-k)}{\overline{n}}<r/\overline{n}.
\end{eqnarray*}
\end{proof}
\begin{Remark}
While the application of the transformation in \cite{eMSR_d_eq_n-1} to cases where $d<n-1$ may seem straightforward, it should be noted that constructing MDS array codes with $d<n-1$ is significantly more challenging. For MDS array codes with $d<n-1$, the node repair process differs substantially from that with $d=n-1$.
For the repair process to succeed, R1 and R2 in Section \ref{sec:repair} should be satisfied for both cases. However, for the case of $d<n-1$, an additional requirement should be fulfilled, i.e., R3 outlined in Section \ref{sec:repair} should be satisfied, or  \eqref{eqn_alternatively_repair} can be solved despite potential interference from unconnected nodes in the equations.
Whereas for the scenario of $d=n-1$, all surviving nodes are connected during the repair of a failed node $i$ ($i\in [0:n)$), resulting in $L_i$ being an empty set in \eqref{Eqn repairment gene}, thereby obviating the need for R3. 

It is noteworthy that determining the necessary conditions for preserving condition R3 or the solvability of \eqref{eqn_alternatively_repair} after the generic transformation is challenging. While we know that $S_{i} A_{t,j}$ is equivalent to $R_{i,j}$ under elementary row operations for a non-compulsory helper node $j$ based on R2, identifying the exact operations that lead to this equivalence is crucial. 
Specifically, we need to find the matrix $B_{t,j,i}$ such that $S_{i} A_{t,j}=B_{t,j,i}R_{i,j}$, which depends on the precise structure of the parity-check matrix and repair matrices. Therefore,
a case-by-case verification of the node repair properties of the resultant code is required, constituting a major challenge in applying the transformation to MSR codes with $d<n-1$.
\end{Remark}

To summarize, by applying the generic transformation in \cite{eMSR_d_eq_n-1} to 
$(\bar{n}, \bar{k})$ MSR codes with a repair degree of $\bar{d}<\bar{n}-1$, we can obtain a new MDS array code with a longer code length. Or equivalently, under the same code length, we can obtain a new MDS array code with a smaller sub-packetization level and a close-to-optimal repair bandwidth, while the number of parity nodes is kept the same. In the following subsections, we apply the generic transformation to the YB codes 1 and 2 in \cite{Barg1}. 

\subsection{An MDS Array Code $\mathcal{C}_2$ With the $(d, d_c)$-Repair Property}\label{sec:C1}
In this subsection, we construct an $(n=s\overline{n},k)$ MDS array code $\mathcal{C}_2$ with the $(d, d_c=s-1)$-repair property by applying the generic transformation to the $(\overline{n}, \overline{k})$ YB code 1 with $\overline{d}<\overline{n}-1$ from \cite{Barg1}. The construction involves carefully selecting $x_{t,j}$ in \eqref{Eqn general coding matrix}  to ensure that it is MDS and R3 holds. Here $r=n-k=\overline{n}-\overline{k}$ and $d<n-1$, with field size being a crucial determinant.


Recall that the $(\overline{n}, \overline{k})$ YB code 1 with $\overline{d}<\overline{n}-1$ from \cite{Barg1} is defined in the form of \eqref{Eqn parity check eq} and \eqref{Eqn A power}, with the optimal update property and the sub-packetization level is $N=w^{\overline{n}}$ where $2\le w=\overline{d}-\overline{k}+1<r$. More precisely, the parity-check matrix $(\overline{A}_{t,i})_{t\in[0:r), i\in[0:\overline{n})}$ of the $(\overline{n}, \overline{k})$ YB code 1 is defined by
\begin{equation}\label{Eqn_YB_code1}
\begin{pmatrix}
V_{i,0}\\
V_{i,1}\\
\vdots\\
V_{i,w-1}
\end{pmatrix}A_{t,i}=
\begin{pmatrix}
\lambda_{i,0}^tV_{i,0}\\
\lambda_{i,1}^tV_{i,1}\\
\vdots\\
\lambda_{i,w-1}^tV_{i,w-1}
\end{pmatrix},
\end{equation}
where $\{\lambda_{i,t}\}_{i\in[0: \overline{n}),t\in[0: w)}$ are $w\overline{n}$ distinct elements in a finite field containing at least $w\overline{n}$ elements.

The
repair matrices and selection matrices of the YB code 1 in \cite{Barg1} are respectively defined by
\begin{eqnarray}\label{Rqn_R_YB1}
\overline{R}_{i,j}=
V_{i,0}+V_{i,1}+\cdots+V_{i,w-1}, i, j\in[0: \overline{n}), j\ne i,
\end{eqnarray}
and
\begin{eqnarray*}\label{Eqn S YB1}
\overline{S}_{i}=
V_{i,0}+V_{i,1}+\cdots+V_{i,w-1}, i\in [0:\overline{n}),
\end{eqnarray*}
where $V_{i,0}, V_{i,1}, \ldots, V_{i,w-1}$ are defined in \eqref{Eqn_Vt}.

\begin{Construction}\label{Con C2}
Let $n=s\overline{n}$ and $k=n-r$. Construct an $(n,k)$ MDS array code $\mathcal{C}_2$ over $\mathbf{F}_q$ by applying the generic transformation in Section \ref{sec:generic tran} to the $(\overline{n}, \overline{k})$ YB code 1 and setting
\begin{eqnarray}\label{Eqn C1 x assi}
x_{t,i}=x_{i}^t
\end{eqnarray}
in \eqref{Eqn general coding matrix} for some $x_i\in \mathbf{F}_q\backslash\{0\}$, where $i\in [0:n)$ and $t\in [0: r)$.
\end{Construction}

The new code $\mathcal{C}_2$ looks similar to the analogous code $\mathcal{C}_1$ in \cite{eMSR_d_eq_n-1}, however, the difference lies in that the former focuses on $d<n-1$ while the latter focuses on $d=n-1$. 
Therefore, we need to cancel additional interference during the node repair process of the new code $\mathcal{C}_2$, since there are some unconnected nodes during this process.

\begin{Theorem}\label{Eqn_C1_field}
The $(n=s\overline{n},k)$ MDS array code $\mathcal{C}_2$ in Construction \ref{Con C2} can be realized over a finite field $\mathbf{F}_q$ with $q>\lceil\frac{s}{w}\rceil w\overline{n}$ and $w|(q-1)$ by setting
\begin{eqnarray}\label{Eqn C1 la assi}
\lambda_{\overline{i},u}=c^{\overline{i}}\delta^{u},~ \overline{i}\in [0: \overline{n}),~ u\in [0: w)
\end{eqnarray} in \eqref{Eqn_YB_code1} and
\begin{equation}\label{Eqn_xi}
x_i=c^{z\overline{n}}\delta^y
\end{equation}
in \eqref{Eqn C1 x assi} for $i=z w \overline{n}+y\overline{n}+\overline{i}\in [0:n)$, $z\in [0: \lceil\frac{s}{w}\rceil)$, $y\in [0:w)$, and $\overline{i}\in [0: \overline{n})$,
where $c$ denotes a primitive element of the finite field $\mathbf{F}_q$ and $\delta=c^{\frac{q-1}{w}}$ denotes a primitive $w$-th root of unity in the finite field $\mathbf{F}_q$.
\end{Theorem}
\begin{proof}
The proofs of the MDS property and the repair property of $\mathcal{C}_2$ are similar to those of $\mathcal{C}_1$ and are omitted here.    
\end{proof}

Given that every block entry $A_{t, i}$ in the parity-check matrix $(A_{t, i})_{t\in [0:r), i\in [0: n)}$ of $\mathcal{C}_2$ is diagonal, we can easily deduce the following result through analysis similar to that of the YB code 1 in \cite{Barg1}.

\begin{Proposition}\label{Prop2}
The code $\mathcal{C}_2$ has the optimal update property.
\end{Proposition}

\subsection{An MDS Array Code $\mathcal{C}_3$ With the $(d,d_c)$-Repair Property}\label{sec:C2}
In this subsection, we construct an $(n=s\overline{n},k)$ MDS array code $\mathcal{C}_3$ with the $(d,d_c=s-1)$-repair property by applying the generic transformation to the $(\overline{n}, \overline{k})$ YB code 2 with $\overline{d}<\overline{n}-1$ from \cite{Barg1}, where $n-k=\overline{n}-\overline{k}$ and $d<n-1$.

The $(\overline{n}, \overline{k})$ YB code 2 with $\overline{d}<\overline{n}-1$ from \cite{Barg1} is defined in the form of \eqref{Eqn parity check eq} and \eqref{Eqn A power} with the sub-packetization level
$N=w^{\overline{n}}$, where $w=\overline{d}-\overline{k}+1$. More precisely,
the parity-check matrix $(\overline{A}_{t,i})_{t\in[0:r), i\in[0:\overline{n})}$ is defined by \begin{equation}\label{Eqn_YB2_CodingM1}
\overline{A}_{t,i}=(\overline{A}_i)^t,~
\left(\begin{array}{c}
     V_{i,0}  \\
     V_{i,1}\\
     \vdots\\
     V_{i,w-1}
\end{array}\right)\overline{A}_i=
\left(\begin{array}{c}
    \lambda_{i} V_{i,1}  \\
     V_{i,2}\\
     \vdots\\
     V_{i,w-1}\\
    V_{i,0}
\end{array}\right), i\in [0:\overline{n})
\end{equation}
where $r=\overline{n}-\overline{k}$ and
\begin{eqnarray}\label{Eqn_lambda_C2}
\lambda_{i}=c^{i+1}
\end{eqnarray}
with $c$ being a primitive element of the finite field $\mathbf{F}_q$ with $q>\overline{n}$.

The
repair matrices and selection matrices of the YB code 2 in \cite{Barg1} are respectively defined by
\begin{eqnarray}\label{Rqn_R_YB2}
\overline{R}_{i,j}=
V_{i,0}, i, j\in[0: \overline{n}), j\ne i,
\end{eqnarray}
and
\begin{eqnarray*}\label{Eqn Liu S}
\overline{S}_{i}=
V_{i,0}, i\in [0: \overline{n}),
\end{eqnarray*}
where $V_{i,0}$ is defined in \eqref{Eqn_Vt}.

By applying the generic transformation in Section \ref{sec:generic tran} to the YB code 2, we construct the desired code $\mathcal{C}_3$ over a small finite field by carefully choosing $x_{t,j}$ in \eqref{Eqn general coding matrix}.

\begin{Construction}\label{Con C3}
Let $n=s\overline{n}$, $k=n-r$, construct an $(n,k)$ array code $\mathcal{C}_3$ by applying the generic transformation in Section \ref{sec:generic tran} to the YB code 2, and setting
\begin{equation}\label{Eqn C2 x assi}
x_{t,i}=x_{i}^t
\end{equation}
in \eqref{Eqn general coding matrix} for some $x_i\in \mathbf{F}_q\backslash\{0\}$, where $i\in [0: n)$ and $t\in [0: r)$.
\end{Construction}

Note that the parity-check matrix of the new code $\mathcal{C}_3$ in Construction \ref{Con C3} is not equivalent to an upper block triangular matrix. Hence, we cannot apply a similar proof to that of the new code $\mathcal{C}_1$. In the following, we will employ the proof method originally used in \cite{Barg1}. The proof proceeds in three steps: i) Create linear equations concerning the data in the failed node, and ii) Null the equations by multiplying both sides with a particular matrix to cancel out the data in the failed node. With the downloaded data, the interference from the unconnected nodes can be obtained by solving these new equations, iii) Solve the equations created in i) since all the interference terms are known now. 

In what follows, we first analyze the structure of the code $\mathcal{C}_3
$.

According to \eqref{Eqn general coding matrix}, \eqref{Eqn_YB2_CodingM1}, and \eqref{Eqn C2 x assi}, the parity-check matrix $(A_{t,i})_{t\in[0:r), i\in[0:n)}$ of the code $\mathcal{C}_3$ is defined by
\begin{equation}\label{Eqn_C3_A_power}
A_{t,i}=A_i^t, i\in [0: n), t\in [0:r),
\end{equation}
with
\begin{equation}\label{Eqn Zigzag d A}
A_i=x_i\overline{A}_{\overline{i}} \mbox{~for~} i\in [0:n) \mbox{~where~} \bar{i}=i\%{\bar{n}}.
\end{equation}

Define
\begin{equation}\label{Eqn_delta}
\delta_u=\left\{\begin{array}{ll} 1, &u=1,\\
0, &u\in [0:w)\backslash\{1\}.\end{array}\right. 
\end{equation}
Clearly, for $i=u\overline{n}+\overline{i}\in [0: n)$, where $u\in [0:s)$ and $\overline{i}\in [0:\overline{n})$, we have
\begin{equation}\label{Eqn zigzag d power A}
\left(\begin{array}{c}
     V_{i,0}  \\
     V_{i,1}\\
     \vdots\\
     V_{i,w-1}
\end{array}\right)A_i^t=x_i^t
\left(\begin{array}{c}
   \lambda_{\bar{i}}^{ \left(\sum\limits_{v=1}^t\delta_{0\oplus v}\right)} V_{i,t\oplus 0}  \\
   \lambda_{\bar{i}}^{\left(\sum\limits_{v=1}^t\delta_{1\oplus v}\right)}  V_{i,t\oplus 1}\\
     \vdots\\
  \lambda_{\bar{i}}^{ \left(\sum\limits_{v=1}^t\delta_{(w-2)\oplus v}\right)}  V_{i,t\oplus (w-2)}\\
  \lambda_{\bar{i}}^{\left(\sum\limits_{v=1}^t\delta_{(w-1)\oplus v}\right)}  V_{i,t\oplus (w-1)}
\end{array}\right)
\end{equation}
for $t\in [0:r)$, where $\oplus$ denotes the addition modulo $w$,
which implies 
\begin{eqnarray}\label{Eqn zigzag d Aw}
A_{i}^w=x_{i}^w \lambda_{\overline{i}}I.
\end{eqnarray}

Before proving the general result of the code $\mathcal{C}_3$, we provide a motivating example below, which conveys the main idea of the general proof.

\begin{Example}
Based on the $(\overline{n}=5, \overline{k}=2)$ YB code 2 with $\overline{d}=3$ and $N=(\overline{d}-\overline{k}+1)^{\overline{n}}=2^5$, from Construction \ref{Con C3} and by setting $s=2$, we obtain an $(n=10, k=7)$ MDS array code $\mathcal{C}_3$ with $N=2^5$ and $d=8$.
Suppose node $4$ fails, then node $9$ is a compulsory helper node that should be contacted. Since $d=8$, one surviving node is not connected, W.l.o.g., say node $0$, is not contacted. We claim that node $4$ can be repaired by  downloading $V_{4,0}\mathbf{f}_j$ (i.e., $f_{j, 2i}$, $i\in [0: 16)$) for $j\in [1: 9)\setminus \{4\}$ and $\mathbf{f}_9$. Below we show how to repair $f_{4,0}$ and $f_{4,1}$, while the remaining symbols $f_{4,j}$, $j\in [2: 32)$ can be regenerated similarly.

By \eqref{Eqn parity check eq} and \eqref{Eqn_C3_A_power}, the code satisfies the following parity-check equations
\begin{equation}\label{Eqn_Zig_ex_pce}
\sum\limits_{i=0}^{9}A_{i}^t\mathbf{f}_{i}=\mathbf{0}, ~t=0,1,2,
\end{equation}
which in conjunction with \eqref{Eqn zigzag d Aw}, we have
\begin{align*}
  & x_4^2c^5\sum\limits_{i=0}^{9}\mathbf{f}_{i} -\sum\limits_{i=0}^{9}A_{i}^2\mathbf{f}_{i}\\=&x_4^2c^5\sum\limits_{i=0}^{9}\mathbf{f}_{i} -\sum\limits_{i=0}^{9}x_{i}^2 \lambda_{i\% 5}\mathbf{f}_{i}\\=&\sum\limits_{i=0}^{3}(x_4^2c^5-x_i^2c^{i+1})\mathbf{f}_{i} +\sum\limits_{i=5}^{9}(x_4^2c^5-x_i^2c^{i-4}) \mathbf{f}_{i}\\=&\mathbf{0}.
\end{align*}

For $m=0, 16$, multiply the above equation by $e_m$ from the left on both sides, we obtain
\begin{equation}\label{Eqn_ex_cancel_interference}
\sum\limits_{i=0}^{3}(x_4^2c^5-x_i^2c^{i+1})f_{i,m} +\sum\limits_{i=5}^{9}(x_4^2c^5-x_i^2c^{i-4}) f_{i,m}=0.    
\end{equation}

From the above equation, $f_{0,0}$ and $f_{0,16}$ (symbols in the unconnected node) can be obtained if $x_4^2c^5-x_0^2c\ne 0$, since the other data in \eqref{Eqn_ex_cancel_interference} can be determined from the downloaded data.

Multiply \eqref{Eqn_Zig_ex_pce} by $e_0$ from the left on both sides, we obtain
the following equations
\begin{equation}\label{Eqn_PCG2-0_0th_bit}
e_0\sum\limits_{i=0}^{9}\mathbf{f}_{i}= \sum\limits_{i=0}^{9}f_{i,0}=0.
\end{equation}
\begin{align} \nonumber &e_0\sum\limits_{i=0}^{9}A_{i}\mathbf{f}_{i}\\\nonumber=&x_0cf_{0,16}+x_1c^2f_{1,8}+x_2c^3f_{2,4}+x_3c^4f_{3,2}+x_4c^5f_{4,1}\\
\nonumber  &+x_5cf_{5,16}+x_6c^2f_{6,8}+x_7c^3f_{7,4}+x_8c^4f_{8,2}+x_9c^5f_{9,1}\\\label{Eqn_PCG2-1_0th_bit}=&0.
\end{align}

With the downloaded data from the helper nodes, the obtained data $f_{0,0}$ and $f_{0,16}$ from the unconnected nodes, 
$f_{4,0}$ and $f_{4,1}$ can be obtained by solving \eqref{Eqn_PCG2-0_0th_bit} and  \eqref{Eqn_PCG2-1_0th_bit}, respectively. The remaining symbols in $\mathbf{f}_4$ can be regenerated similarly.
\end{Example}

From now on, w.l.o.g., we consider the repair of node $u\overline{n}+i$ for any given $u\in [0: s)$ and $i\in [0: \overline{n})$.

For $v\in [0: s)$ and $j\in [0: \bar{n})\backslash\{\bar{i}\}$,
define the $N/w\times N/w$ matrix $B_{v\overline{n}+j,i}$ as
\begin{equation}\label{Eqn_B_le_i}
\left(\begin{array}{c}
     \hat{V}_{j,0}  \\
     \hat{V}_{j,1}\\
     \vdots\\
     \hat{V}_{j,w-1}
\end{array}\right)B_{v\overline{n}+j,i}=x_{v\overline{n}+j}
\left(\begin{array}{c}
    \lambda_{j} \hat{V}_{j,1}  \\
     \hat{V}_{j,2}\\
     \vdots\\
     \hat{V}_{j,w-1}\\
    \hat{V}_{j,0}
\end{array}\right) \mbox{~for~} j<i
\end{equation} and
\begin{equation}\label{Eqn_B_gt_i}
\left(\begin{array}{c}
     \hat{V}_{j-1,0}  \\
     \hat{V}_{j-1,1}\\
     \vdots\\
     \hat{V}_{j-1,w-1}
\end{array}\right)B_{v\overline{n}+j,i}=x_{v\overline{n}+j}
\left(\begin{array}{c}
    \lambda_{j} \hat{V}_{j-1,1}  \\
     \hat{V}_{j-1,2}\\
     \vdots\\
     \hat{V}_{j-1,w-1}\\
    \hat{V}_{j-1,0}
\end{array}\right)\mbox{~for~} j>i,
\end{equation}
where $\hat{V}_{j,u}$ and $\hat{V}_{j-1,u}$ are defined in Section \ref{sec:partition}.
Then we have the following results.

\begin{Lemma}\label{le_BB}
For any $i\in [0: \bar{n})$,  $v, v'\in [0: s)$, and $j,j'\in [0:\overline{n})\backslash\{i\}$ with $(v, j)\ne (v', j')$, we have $$B_{v\overline{n}+j,i}B_{v'\overline{n}+j',i}=B_{v'\overline{n}+j',i}B_{v\overline{n}+j,i},$$ and  if
\begin{itemize}
    \item [i)] 
$    
x_{v\overline{n}+j}^w\lambda_{j}\ne x_{v'\overline{n}+j'}^w\lambda_{j'}
$ when $j\ne j'$.
\item [ii)] $x_{v\overline{n}+j}\ne x_{v'\overline{n}+j}$ when $v\ne v'$,
\end{itemize}
we further have that $B_{v\overline{n}+j,i}-B_{v'\overline{n}+j',i}$ is non-singular.
\end{Lemma}
\begin{proof}
The proof is given in Appendix \ref{sec:pfBBcomute}.
\end{proof}

\begin{Lemma}\label{lem VAf}
For $i,j\in [0:\overline{n})$ with $i\ne j$, $v\in [0:s)$, and $m\in [0:r-w)$, we have
\begin{eqnarray*}
V_{i,0}A_{v\overline{n}+j}^m=
B_{v\overline{n}+j,i}^mV_{i,0}.
\end{eqnarray*}
\end{Lemma}
\begin{proof}
The proof is given in Appendix \ref{sec:pfC2}.
\end{proof}

\begin{Theorem}\label{Thm_C2}
The code $\mathcal{C}_3$ in Construction \ref{Con C3} is an $(n=s\overline{n}, k)$ MDS array code over $\mathbf{F}_q$ with the $(d=n-\overline{n}+\overline{d}, d_c=s-1)$ repair property and $(1+\frac{d_c(d-k)}{d})$-optimal repair bandwidth
if i) and ii) of Lemma \ref{le_BB} are satisfied.
\end{Theorem}
\begin{proof}
Similar to the proof of Lemma \ref{le_BB}, we can derive that $A_{v\overline{n}+j}A_{v'\overline{n}+j'}=A_{v'\overline{n}+j'}A_{v\overline{n}+j}$ and $A_{v\overline{n}+j}-A_{v'\overline{n}+j'}$ is non-singular
if
\begin{equation*}\label{Eqn_C2_con1}
x_{v\overline{n}+j}^w\lambda_{j}\ne x_{v'\overline{n}+j'}^w\lambda_{j'}
\end{equation*}
holds
for all $v,v'\in [0: s)$ and $j, j'\in [0: \overline{n})$ with $(v, j)\ne (v', j')$, which can be satisfied if i) and ii) of Lemma \ref{le_BB} hold. Therefore, according to Lemma \ref{Lemma block matrix}, we conclude that
the code $\mathcal{C}_3$ defined in Construction \ref{Con C3} possesses the MDS property if i) and ii) of Lemma \ref{le_BB} are met.

To verify the repair property, we only need to demonstrate that R3 holds for the code $\mathcal{C}_3$ based on Theorem \ref{Thm general repair} since R1 and R2 hold for YB code 2 \cite{Barg1}. Assuming node $u\overline{n}+i$ fails, where $i\in [0:\overline{n}),~u\in [0:s)$, we connect node $v\overline{n}+i$ for all $v\in [0:s)\backslash\{u\}$ and any other $d-d_c$ surviving nodes during the repair process, and download $R_{u\overline{n}+i,v\overline{n}+j}\mathbf{f}_{v\overline{n}+j}$ from node $v\overline{n}+j$ that is connected. By \eqref{Eqn general R} and \eqref{Rqn_R_YB2}, we have $R_{u\overline{n}+i,v\overline{n}+j}=I$ if $j=i$ and $R_{u\overline{n}+i,v\overline{n}+j}=\overline{R}_{i,j}=V_{i,0}$ otherwise. That is, we download $\mathbf{f}_{v\overline{n}+i}$ for all $v\in [0:s)\backslash\{u\}$ and $V_{i,0}\mathbf{f}_{v\overline{n}+j}$ for any $d-d_c$ distinct pairs of $(v, j)$ with $v\in [0:s), j\in [0:\overline{n})\backslash\{i\}$. To prove that R3 holds for the code $\mathcal{C}_3$, it suffices to show that with the data $\mathbf{f}_{v\overline{n}+i}$, $v\in [0:s)\backslash\{u\}$, any $d-d_c=w+k-s$ elements from the set $$\{V_{i,0}\mathbf{f}_{v\overline{n}+j}|v\in [0:s), j\in [0:\overline{n})\backslash\{i\}\}$$ can recover the remaining $r-w$ elements in the set.

For $m\in [0:r-w)$, by \eqref{Eqn parity check eq}, \eqref{Eqn A power}, \eqref{Eqn Zigzag d A}, and \eqref{Eqn zigzag d Aw}, we derive
\begin{equation}\label{Eqn_1st_eq}
\sum\limits_{j=0}^{n-1}A_j^m\mathbf{f}_j=\sum\limits_{v=0}^{s-1} \sum\limits_{j=0}^{\overline{n}-1}A_{v\overline{n}+j}^m\mathbf{f}_{v\overline{n}+j}=\bf{0},
\end{equation}
and
\begin{align}\label{Eqn_2st_eq}
\nonumber \sum\limits_{j=0}^{n-1}A_j^{m+w}\mathbf{f}_j&=\sum\limits_{v=0}^{s-1} \sum\limits_{j=0}^{\overline{n}-1}A_{v\overline{n}+j}^{m+w}\mathbf{f}_{v\overline{n}+j}\\\nonumber&=\sum\limits_{v=0}^{s-1} \sum\limits_{j=0}^{\overline{n}-1}x_{v\overline{n}+j}^w\lambda_{j}A_{v\overline{n}+j}^{m}\mathbf{f}_{v\overline{n}+j}\\
&=\bf{0}.
\end{align}
By multiplying \eqref{Eqn_1st_eq} with $x_{u\overline{n}+i}^w\lambda_{i}$ on the left and subtracting it from \eqref{Eqn_2st_eq}, we obtain
\begin{equation*}
\sum\limits_{v=0}^{s-1} \sum\limits_{j=0}^{\overline{n}-1}\left(x_{v\overline{n}+j}^w\lambda_{j}-x_{u\overline{n}+i}^w\lambda_{i}\right)A_{v\overline{n}+j}^m\mathbf{f}_{v\overline{n}+j}=\bf{0},
\end{equation*}
which implies
\begin{align*}
&\sum\limits_{v=0}^{s-1} \sum\limits_{j=0,j\ne i}^{\overline{n}-1}(x_{v\overline{n}+j}^w\lambda_{j}-x_{u\overline{n}+i}^w\lambda_{i})A_{v\overline{n}+j}^m\mathbf{f}_{v\overline{n}+j}\\=&-\lambda_{i}\sum\limits_{v=0,v\ne u}^{s-1}(x_{v\overline{n}+i}^w-x_{u\overline{n}+i}^w)A_{v\overline{n}+i}^m\mathbf{f}_{v\overline{n}+i}\\=&\kappa_{*},
\end{align*}
where $\kappa_{*}$ denotes some known data since $\mathbf{f}_{v\overline{n}+i}$ is downloaded for all $v\in [0:s)\backslash\{u\}$.
By multiplying both sides of the equation by $S_{u\overline{n}+i}=\overline{S}_i=V_{i,0}$, we obtain
\begin{equation*}
\sum\limits_{v=0}^{s-1} \sum\limits_{j=0,j\ne i}^{\overline{n}-1}(x_{v\overline{n}+j}^w\lambda_{j}-x_{u\overline{n}+i}^w\lambda_{i})V_{i,0}A_{v\overline{n}+j}^m\mathbf{f}_{v\overline{n}+j}=V_{i,0}\kappa_{*},
\end{equation*}
which in conjunction with Lemma \ref{lem VAf} leads to
\begin{align*}
&\sum\limits_{v=0}^{s-1}\left(\sum\limits_{j=0}^{i-1}(x_{v\overline{n}+j}^w\lambda_{j}-x_{u\overline{n}+i}^w\lambda_{i})B_{v\overline{n}+j,i}^mV_{i,0}\mathbf{f}_{v\overline{n}+j}\right.\\
&\left.+\sum\limits_{j=i+1}^{\overline{n}-1}(x_{v\overline{n}+j}^w\lambda_{j}-x_{u\overline{n}+i}^w\lambda_{i})B_{v\overline{n}+j,i}^m V_{i,0}\mathbf{f}_{v\overline{n}+j}\right)=V_{i,0}\kappa_{*},
\end{align*}
for $m\in [0:r-w)$,
which can be expressed as in \eqref{Eqn zigzag d reduced eq},
\begin{figure*}
\begin{eqnarray}\label{Eqn zigzag d reduced eq}
\underbrace{
\setlength{\arraycolsep}{1pt}\begin{pmatrix}
(x_0^w\lambda_0-x_{u\overline{n}+i}^w\lambda_{i})I & (x_1^w\lambda_1-x_{u\overline{n}+i}^w\lambda_{i})I & \cdots & (x_{s\overline{n}-1}^w\lambda_{\overline{n}-1}-x_{u\overline{n}+i}^w\lambda_{i})I \\
(x_0^w\lambda_0-x_{u\overline{n}+i}^w\lambda_{i})B_{0,i} & (x_1^w\lambda_1-x_{u\overline{n}+i}^w\lambda_{i})B_{1,i} & \cdots & (x_{s\overline{n}-1}^w\lambda_{\overline{n}-1}-x_{u\overline{n}+i}^w\lambda_{i})B_{s\overline{n}-1,i} \\
\vdots & \vdots & \ddots & \vdots \\
(x_0^w\lambda_0-x_{u\overline{n}+i}^w\lambda_{i})B_{0,i}^{r-w-1} & (x_1^w\lambda_1-x_{u\overline{n}+i}^w\lambda_{i})B_{1,i}^{r-w-1} & \cdots & (x_{s\overline{n}-1}^w\lambda_{\overline{n}-1}-x_{u\overline{n}+i}^w\lambda_{i})B_{s\overline{n}-1,i}^{r-w-1} \\
\end{pmatrix}}_{M}
\underbrace{\setlength{\arraycolsep}{1pt}\begin{pmatrix}
V_{i,0}\mathbf{f}_{0}\\ V_{i,0}\mathbf{f}_{1}\\\vdots\\V_{i,0}\mathbf{f}_{s\overline{n}-1}
\end{pmatrix}}_{F}=V_{i,0}\kappa_{*}
\end{eqnarray}
    \hrulefill
\end{figure*}
where $M$ in \eqref{Eqn zigzag d reduced eq} is a $(r-w)\times s(\bar{n}-1)$ block matrix. For $v\in [0: s)$, block column $v(\overline{n}-1)+j$ of $M$ and row $v(\overline{n}-1)+j$ of $F$ are
\begin{equation*}
\begin{pmatrix}
(x_{v\overline{n}+j}^w\lambda_{j}-x_{u\overline{n}+i}^w\lambda_{i})I\\
(x_{v\overline{n}+j}^w\lambda_{j}-x_{u\overline{n}+i}^w\lambda_{i})B_{v\overline{n}+j,i}\\
\vdots\\
(x_{v\overline{n}+j}^w\lambda_{j}-x_{u\overline{n}+i}^w\lambda_{i})B_{v\overline{n}+j,i}^{r-w-1}
\end{pmatrix}\mbox{~and~}V_{i,0}\mathbf{f}_{v\overline{n}+j}
\end{equation*}
if $j<i$, and are
\begin{equation*}
\begin{pmatrix}
(x_{v\overline{n}+j+1}^w\lambda_{j+1}-x_{u\overline{n}+i}^w\lambda_{i})I\\
(x_{v\overline{n}+j+1}^w\lambda_{j+1}-x_{u\overline{n}+i}^w\lambda_{i})B_{v\overline{n}+j+1,i}\\
\vdots\\
(x_{v\overline{n}+j+1}^w\lambda_{j+1}-x_{u\overline{n}+i}^w\lambda_{i})B_{v\overline{n}+j+1,i}^{r-w-1}
\end{pmatrix} \mbox{~and~}V_{i,0}\mathbf{f}_{v\overline{n}+j+1}
\end{equation*}
if $i\le j<\bar{n}-1$. Note that $M$ in \eqref{Eqn zigzag d reduced eq} can be rewritten as
\begin{equation*}
M=M_0M_1
\end{equation*}
where
\begin{equation*}
M_0=\begin{pmatrix}
I & I & \cdots & I \\
B_{0,i} & B_{1,i} & \cdots & B_{s\overline{n}-1,i} \\
\vdots & \vdots & \ddots & \vdots \\
B_{0,i}^{r-w-1} & B_{1,i}^{r-w-1} & \cdots & B_{s\overline{n}-1,i}^{r-w-1}
\end{pmatrix},
\end{equation*}
\begin{equation*}
M_1=\setlength{\arraycolsep}{1pt}\begin{pmatrix}
(x_0^w\lambda_0-x_{u\overline{n}+i}^w\lambda_{i})I  & & \\
 & \ddots & \\ & & (x_{s\overline{n}-1}^w\lambda_{\overline{n}-1}-x_{u\overline{n}+i}^w\lambda_{i})I
\end{pmatrix}.
\end{equation*}

By i) and ii) of Lemma \ref{le_BB}, it follows that $M_1$ is non-singular and any $(r-w)\times (r-w)$ sub-block matrix of the matrix
$M_0$
is non-singular according to Lemma \ref{Lemma block matrix}, i.e., any $(r-w)\times (r-w)$ sub-block matrix of $M$ in \eqref{Eqn zigzag d reduced eq}
is of full rank. Therefore, with the data $\mathbf{f}_{v\overline{n}+i}$, $v\in [0:s)\backslash\{u\}$,
any $d-d_c=w+k-s$ elements from the set $\{V_{i,0}\mathbf{f}_{v\overline{n}+j}|v\in [0: s), j\in [0: \overline{n})\backslash\{i\}\}$ can recover the remaining $r-w$ elements by solving \eqref{Eqn zigzag d reduced eq}. This completes the proof.
\end{proof}

\begin{Theorem}\label{Thm_C2_field}
The requirements in items i) and ii) of Lemma \ref{le_BB} can be satisfied by setting $x_{v\overline{n}+i}=c^{v\lceil\frac{\overline{n}}{w}\rceil}$ for $v\in [0: s)$ and $i\in [0:\overline{n})$, where $c$ is the primitive element of
$\mathbf{F}_q$ with $q>\lceil\frac{\overline{n}}{w}\rceil sw$.
\end{Theorem}

\begin{proof}
For $v, v'\in [0: s)$ and $j, j'\in [0: \bar{n})$, by \eqref{Eqn_lambda_C2}, we have
\begin{itemize}
    \item [i)] If $j\ne j'$, we have
\begin{equation*}
\frac{x_{v\overline{n}+j}^w\lambda_{j}}{x_{v'\overline{n}+j'}^w\lambda_{j'}}=c^{w\lceil\frac{\overline{n}}{w}\rceil(v-v')+j-j'}\ne 1
\end{equation*}
since
\begin{align*}
 0&<\left|w\left\lceil\frac{\overline{n}}{w}\right\rceil(v-v')+j-j'\right|\\
 &\le w\left\lceil\frac{\overline{n}}{w}\right\rceil(s-1)+\overline{n}-1\\
 &\le \left\lceil\frac{\overline{n}}{w}\right\rceil sw+\bar{n}-w\left\lceil\frac{\overline{n}}{w}\right\rceil-1\\
 &<q-1.
\end{align*}

\item [ii)]  If $v\ne v'$, then it is obvious that $x_{v\overline{n}+j}\ne x_{v'\overline{n}+j}$.
\end{itemize}
This completes the proof.
\end{proof}

\begin{Remark}
By \eqref{Eqn general R} and \eqref{Rqn_R_YB2}, it is evident that the new MDS array code $\mathcal{C}_3$ possesses the $(1+\epsilon)$-optimal access property, where $\epsilon=\frac{d_c(d-k)}{d}$.
\end{Remark}

\section{Comparisons}\label{sec:comp}
In this section, we compare key parameters among the proposed MDS array codes and some existing notable MDS array codes, as illustrated in Table \ref{Table comp}.

\begin{table*}[htbp]
\begin{center}
\caption{A comparison of some key parameters among the $(n=k+r,k)$ MDS array codes with $d<n-1$ and $(1+\epsilon)$-optimal repair bandwidth proposed in this paper and some existing notable ones, where $\epsilon=\frac{d_c(d-k)}{d}$, $w=d-k+1$, $N_{ag}$ and $K_{ag}$ denote the code length and code dimension of the algebraic geometry (AG) code used in the MDS array code in \cite{Guruswami2020}, $\lambda$, $\delta$, and $t$ denote the length, the relative minimum distance, the alphabet size of a code with size $n$, and $\overline{n}$ denotes the code length of the underlying MSR code employed}\label{Table comp}
\setlength{\tabcolsep}{.1pt}
\begin{tabular}{|c|c|c|c|c|c|c|c|}
\hline &
Code
&\multirow{2}{*}{$N$}& \multirow{2}{*}{Field size } &The ratio of repair bandwidth &Number of compulsory& \multirow{2}{*}{Remark}\\
&length $n$&  & &to the optimal value $\gamma_{\rm optimal}$ &  helper nodes & \\
\hline
New code $\mathcal{C}_1$ & $2ms$ & $w^{m}$ & $q>2mw\left\lceil s/w\right\rceil$ & $=1+\frac{d_c(d-k)}{d}<1+r/\overline{n}$ & $s-1$ & \\
\hline
The MSR code in \cite{li2023MDS} & $n$ & $w^{\lceil n/2\rceil}$ &
$q>\hspace{-1mm}\left\{\hspace{-2mm}\begin{array}{cc}
    4\lceil n/2\rceil ,& w=2  \\
    (w+1) \lceil n/2\rceil,& 3\le w<r
\end{array}\right.$ & $1$~(Optimal) & 0& \\
\hline
\hline
New code $\mathcal{C}_2$ & $s\overline{n}$ & $w^{\overline{n}}$ & $q>\lceil\frac{s}{w}\rceil w\overline{n}$, $w|(q-1)$ & $=1+\frac{d_c(d-k)}{d}<1+r/\overline{n}$& $s-1$ & Optimal update 
 \\
\hline
YB code 1 \cite{Barg1}& $n$ &$w^{n}$ & $q\ge wn$ & $1$~(Optimal)& $0$ & Optimal update \\
\hline\hline
New code $\mathcal{C}_3$ & $s\overline{n}$ & $w^{\overline{n}}$ & $q> \lceil\frac{\overline{n}}{w}\rceil sw$ & $=1+\frac{d_c(d-k)}{d}<1+r/\overline{n}$& $s-1$ & $(1+\epsilon)$-optimal access \\
\hline
YB code 2 \cite{Barg1} & $n$ &$w^{n}$ & $q>n$ & $1$~(Optimal)& $0$ & Optimal access \\
\hline
\hline
\multirow{3}{*}{GLJ code \cite{Guruswami2020}}& \multirow{3}{*}{$\overline{n}^{K_{ag}}$} & \multirow{3}{*}{$N_{ag}w^{\overline{n}}$} & \multirow{3}{*}{$q>r\overline{n}^{K_{ag}+1}$} & \multirow{3}{*}{$\le 1+\epsilon$} &\multirow{3}{*}{$\approx n-n^{\frac{u-1}{u}}$} &$\overline{n}>2(u+1)^2\frac{r^2}{\epsilon^2}$ \\ &&&&&&$u\ge 4$ \\
& &&&&&$\overline{n}$ is a square prime power \\
\hline
RRT code \cite{Ramkumar2025MSR} & $n$ &$\lambda w^{t}$ & $q>wn$ & $\le1+(1-\delta)(w-1)$& $0$ &  \\
\hline
\end{tabular}
\end{center}
\end{table*}

\begin{table*}[htbp]
\begin{center}
\caption{A numerical comparison of the (smallest) supported code length and sub-packetization level between the $(n=k+r,k)$ MDS array code $\mathcal{C}_1$ with $d<n-1$ proposed in this paper and the $(n,k)$ GLJ code in \cite{Guruswami2020}, where $\gamma$, $\gamma_{\rm optimal}$, and $\frac{d_c}{n-1}$ denote the repair bandwidth, the value of the optimal repair bandwidth, and the ratio of the number of compulsory helper nodes to the number of surviving nodes, respectively. To ensure a fair comparison, we set $d=k+1$ for $\mathcal{C}_1$, similar to the examples of the code parameters in \cite{Guruswami2020}. In addition, the shortening technique is deployed on the code $\mathcal{C}_1$ when $n$ cannot be factored into $2ms$ for 
 desired integers $m$ and $s$.}\label{Table comp_numerical}
\setlength{\tabcolsep}{2.5pt}
\begin{tabular}{|c|c|c|c|c|c|c|c|}
\hline &$\frac{\gamma}{\gamma_{\rm optimal}}$ (i.e., $1+\epsilon$)& Code length $n$
& $N$ & $\frac{d_c}{n-1}$ & Field size $q$ & Remark\\
\hline
GLJ code ($r=3$) & $\le1+0.7$ & $961^{12}$ &$90\times 2^{961}$ & $\le9.37\%$ & $>3\times 961^{13}$ & Tables I and II in \cite{Guruswami2020}\\\hline
New code $\mathcal{C}_1$ ($r=3$) & $\le1+0.1667$ & $8$ & $2^2$ & $\le 14.29\%$& $>8$ & Set $s=w=2$, $m=2$ \\\hline
New code $\mathcal{C}_1$ ($r=3$) & $\le1+0.0919$ & $100$ & $2^5$& $\le 9.1\%$& $>100$ &Set $s=10$, $w=2$, $m=5$ \\ 
\hline
New code $\mathcal{C}_1$ ($r=3$) & $\le1+0.0162$ & $961^{12}$ & $2^{31}$& $\le 1.61\%$& $>961^{12}+31$ & Set $s=\frac{31\times 961^{11}+1}{2}$, $w=2$, $m=31$ \\
\hline
New code $\mathcal{C}_1$ ($r=3$) & $\le1+0.1667$ & $961^{12}$ & $2^{3}$& $\le 16.67\%$& $>961^{12}+11$ & Set $s=\lceil \frac{961^{12}}{6} \rceil$, $w=2$, $m=3$ \\
\hline
\hline
GLJ code ($r=4$) & $\le1+0.8$ & $1369^{12}$ &$108\times 2^{1369}$ & $\le7.89\%$ & $>4\times 1369^{13}$ & Tables I and II in \cite{Guruswami2020} \\\hline
New code $\mathcal{C}_1$ ($r=4$) & $\le1+0.1112$ & $12$ & $2^3$& $\le 9.1\%$ & $>12$ & Set $s=w=2$, $m=3$ \\ \hline
New code $\mathcal{C}_1$ ($r=4$) & $\le1+0.0961$ & $180$ & $2^5$& $\le 9.5\%$ & $>180$ & Set $s=18$, $w=2$, $m=5$ \\
\hline
New code $\mathcal{C}_1$ ($r=4$) & $\le1+0.0135$ &$1369^{12}$ & $2^{37}$& $\le 1.35\%$ & $>1369^{12}+111$ & Set $s=\frac{37\cdot 1369^{11}+1}{2}$, $w=2$, $m=37$ \\
\hline
New code $\mathcal{C}_1$ ($r=4$) & $\le1+0.1667$ &$1369^{12}$ & $2^{3}$& $\le 16.67\%$ & $>1369^{12}+11$ & Set $s=\lceil \frac{1369^{12}}{6} \rceil$, $w=2$, $m=3$ \\
\hline
\end{tabular}
\end{center}
\end{table*}

{Table \ref{Table comp} shows that the new MDS array codes have significantly smaller sub-packetization levels than the original MSR codes. Additionally, compared to the GLJ code in \cite{Guruswami2020}, the proposed MDS array codes offer the following advantages:

\begin{itemize}
\item The new MDS array codes $\mathcal{C}_1$, $\mathcal{C}_2$, and $\mathcal{C}_3$ derived in this paper are valid for any code length $n\ge 8$. In contrast, the GLJ code in \cite{Guruswami2020} is limited to very large code lengths according to Theorem III.10, Table I, and Appendix F in \cite{Guruswami2020}, typically $n=(\overline{n})^{K_{ag}}>200^{K_{ag}}$ with $\overline{n}$ being a prime power larger than $2(u+1)^2\frac{r^2}{\epsilon^2}\ge 200$, $u\ge 4$ and $\epsilon<1$, where $K_{ag}$ denotes the code dimension of a specific AG code. For example, from Tables I and II in \cite{Guruswami2020}, when $r=3$ and $\gamma\le 1.7\gamma_{\rm optimal}$, the code length and sub-packetization level of GLJ code are at least $961^{12}$ and $90\times 2^{961}$, respectively. Whereas the new MDS array codes $\mathcal{C}_1$-$\mathcal{C}_3$ can support any code length $n\ge8$, with sub-packetization levels of $2^2$ for $\mathcal{C}_1$ and $2^5$ for $\mathcal{C}_2$ and $\mathcal{C}_3$.

\item The new MDS array codes have significantly smaller sub-packetization levels $N$ than the GLJ code in \cite{Guruswami2020}. More specifically, the sub-packetization levels of the new MDS array codes $\mathcal{C}_1$, $\mathcal{C}_2$, and $\mathcal{C}_3$ are $N =w^{\overline{n}/2}$, $w^{\overline{n}}$, and $w^{\overline{n}}$, respectively, which can be a constant and independent of code length $n$ once $w$ and $\overline{n}$ are fixed. On the other hand, the sub-packetization level $N$ of the GLJ code in \cite{Guruswami2020} requires $N=N_{ag}w^{\overline{n}}>N_{ag}w^{200}$ as $\overline{n}>200$, where $N_{ag}$ denotes the code length of a specific AG code, and $N_{ag}\ge 17$ since $N_{ag}$ is a multiple of $\sqrt{\overline{n}}-1$ according to Theorem III.5 in \cite{Guruswami2020}.

\item The required finite field sizes for all the new MDS array codes $\mathcal{C}_1$, $\mathcal{C}_2$, and $\mathcal{C}_3$ are at most of the order $O(wn)$, and are comparable to their code lengths $n$ when $w\le s$ or $s\%w$ is close to $w$ for $\mathcal{C}_1$ and $\mathcal{C}_2$, and $w\le \bar{n}$ or $\bar{n}\%w$ is close to $w$ for $\mathcal{C}_3$,  which are smaller than that of the GLJ code in \cite{Guruswami2020} (i.e., $O(r\bar{n}n)$), making them more practical for implementation.

\item The number $d_c$ of compulsory helper nodes for the GLJ code in \cite{Guruswami2020} is roughly $n-n^{\frac{u-1}{u}}$, meaning most surviving nodes must act as helper nodes for large $n$. In contrast, the new codes $\mathcal{C}_1$-$\mathcal{C}_3$ require only $s-1$ nodes to be compulsory helper nodes, which is only a small portion.

\item Furthermore, the new MDS array code $\mathcal{C}_2$ exhibits the optimal update property while the new MDS array code $\mathcal{C}_3$ possesses the $(1+\epsilon)$-optimal access property.
\end{itemize}

Independent and parallel to our work, Ramkumar, Raviv, and Tamo developed new $\epsilon$-MSR codes in \cite{Ramkumar2025MSR}, which we refer to as the RRT code in this paper.  However, the proposed codes rely on the existence of a code over the alphabet $[1: t+1)$  with length $\lambda$, size $n$, and minimum distance at least $\lambda\delta$, which restricts the code parameters they can achieve.

In addition, we present a numerical comparison in Table \ref{Table comp_numerical}. Since $\mathcal{C}_1$ has the smallest sub-packetization level among all the new MDS array codes, we only list the smallest
code length that the GLJ code in \cite{Guruswami2020} and the new MDS array code $\mathcal{C}_1$ can support, along with other parameters for repair degree $d=k+1$ and $r=3, 4$. The parameters of the new MDS array codes $\mathcal{C}_2$ and $\mathcal{C}_3$ can be similarly derived. Additionally, we list some other code lengths that the new MDS array code $\mathcal{C}_1$ can support. From Table \ref{Table comp_numerical} we have the following observations:
\begin{itemize}
\item For the GLJ code in \cite{Guruswami2020} with $3$ parity nodes, when the repair bandwidth is no more than $(1+0.7)$ times the optimal value, the smallest
code length is $961^{12}$, while the sub-packetization level is $90\times 2^{961}$. In contrast, for the new MDS array code $\mathcal{C}_1$ with $3$ parity nodes, the smallest
supported code length is $8$, with the repair bandwidth only $(1+0.1667)$ times the optimal value and a sub-packetization level of $2^2$.


\item  
Care is needed in comparing compulsory helpers. The GLJ estimate $d_c\le n-n^{(u-1)/u}$ is an upper bound, so it does not establish that almost all surviving nodes are compulsory. A further bound can be derived directly from its repair rule. If the outer linear code has length $N_{ag}$, size $n$, and alphabet size $\overline n$, with a surjective map at each coordinate, then, for any failed index word, exactly $n/\overline n-1$ other codewords agree with it in a given coordinate. Since the compulsory set is the union of these coordinate-agreement sets,
\[
d_c\le N_{ag}\left(\frac n{\overline n}-1\right).
\]
For the AG evaluation codes in the compared examples in Table \ref{Table comp_numerical}, the coordinate maps are surjective because the code contains the constant words. The resulting upper proportions are below $90/961$ and $108/1369$, respectively. Table \ref{Table comp_numerical} reports these derived bounds, not exact compulsory-set sizes.
\end{itemize}

\section{Conclusion}\label{sec:conclusion}
In this paper, three explicit MDS array codes with small sub-packetization levels and $(1+\epsilon)$-optimal repair bandwidth were obtained. All the new MDS array codes operate over small finite fields with sizes comparable to their code lengths, and work for a repair degree of $d<n-1$. Comparisons showed that the obtained MDS array codes outperform
existing ones in terms of parameter flexibility, field size, and sub-packetization
level. 
Removing the compulsory nodes of the new codes in this work will be part of our ongoing work.

\appendices

\section{Proof of Proposition \ref{prop:SA_C1}}\label{sec:pfC3}
We prove \eqref{Eqn the first eq for proving Thm 9} according to the following cases.
\begin{itemize}
\item [] Case 1: For $i,j\in [0: ms)$ and $i\equiv j \bmod m$, by \eqref{Eqn_PC_New1}, \eqref{Eqn  select matrix of code New1}, and \eqref{Eqn_Vt2}, we have
\begin{align*}
S_{i}A_{t,j}=V_{j,0}A_{t,j}=\lambda_{j,0}^t V_{i,0}+\sum\limits_{u=1}^{w-1}(\lambda_{j,0}^t-\lambda_{j,u}^t )V_{i,u}.    
\end{align*}

 \item [] Case 2:  For $i,j\in [ms: n)$ and $i\equiv j \bmod m$, by \eqref{Eqn_PC_New1}, \eqref{Eqn select matrix of code New1}, and \eqref{Eqn_Vt2}, we have
\begin{align*}
S_{i}A_{t,j}&=(V_{i,0}+V_{i,1}+\cdots+V_{i,w-1})A_{t,j}\\
&=\lambda_{j,0}^t V_{i,0}+\lambda_{j,1}^t V_{i,1}+\cdots+\lambda_{j,w-1}^t V_{i,w-1}.    
\end{align*}

 \item [] Case 3a: For $i\in [0: ms)$ and $j\in [ms: n)$ with $i\equiv j \bmod m$, by \eqref{Eqn_PC_New1}, \eqref{Eqn select matrix of code New1}, \eqref{Eqn repair matrix of code New1}, and \eqref{Eqn_Vt2} we have
\begin{align*}
S_{i}A_{t,j}=V_{j,0}A_{t,j}=\lambda_{j,0}^t V_{i,0}=\lambda_{j,0}^tR_{i,j}.    
\end{align*}

 \item [] Case 3b: For $i\in [ms: n)$ and $j\in[0: ms)$ with $i\equiv j \bmod m$, by \eqref{Eqn_PC_New1}, \eqref{Eqn select matrix of code New1}, \eqref{Eqn repair matrix of code New1}, and \eqref{Eqn_Vt2}, we have
\begin{align*}
S_{i}A_{t,j}&=(V_{j,0}+V_{j,1}+\cdots+V_{j,w-1})A_{t,j}\\
&=\lambda_{j,0}^t (V_{j,0}+V_{j,1}+\cdots+V_{j,w-1})\\
&=\lambda_{j,0}^tR_{i,j}.    
\end{align*}

 \item [] Case 4a: For $i\in [0:n)$ and $j\in [0: ms)$ with $i\not\equiv j \bmod m$, we first consider that $i\in [0:m)$. Let $\bar{j}\in [0: m)$ such that $\bar{j}\equiv j \bmod m$, for $a\in [0: N/w)$, by \eqref{Eqn_PC_New1}, we have

\begin{align*}
&V_{i,u}[a:]A_{t,j}\\=&e_{g_{i,u}(a)}A_{t,j}\\=&\left\{\hspace{-2mm}\begin{array}{ll}
 \lambda_{j,0}^t e_{g_{i,u}(a)}& \textrm{if~}(g_{i,u}(a))_{\bar{j}}=0\\
 +\hspace{-1mm}\sum\limits_{v=1}^{w-1}   (\lambda_{j,0}^t\hspace{-1mm}-\hspace{-1mm}\lambda_{j,v}^t) e_{(g_{i,u}(a))(\bar{j},v)},&\\
 \lambda_{j,(g_{i,u}(a))_{\bar{j}}}^t e_{g_{i,u}(a)}, &\textrm{if~}(g_{i,u}(a))_{\bar{j}}\ne 0.
\end{array}\right.
\end{align*}
\begin{itemize}
    \item  If $\bar{j}<i$, by \eqref{Eqn_B<} and \eqref{Eqn e_ae_bT}, we have
\begin{align*}
&B_{t,j,i}[a:]V_{i,u}\\=&\left\{\hspace{-2.1mm}\begin{array}{ll}
 (\lambda_{j,0}^t \hat{e}_{a}+\sum\limits_{v=1}^{w-1}   (\lambda_{j,0}^t-\lambda_{j,v}^t) \hat{e}_{a(\bar{j},v)})& \textrm{if~}a_{\bar{j}}=0,\\
\hspace{2cm} \times\sum\limits_{b=0}^{N/w-1}\hat{e}_b^{\top} e_{g_{i,u}(b)},&\\
 \lambda_{j,a_{\bar{j}}}^t \hat{e}_{a}\sum\limits_{b=0}^{N/w-1}\hat{e}_b^{\top} e_{g_{i,u}(b)}, &\textrm{if~}a_{\bar{j}}\ne 0,
\end{array}\right.\\
=&\left\{\begin{array}{ll}
\sum\limits_{v=1}^{w-1}(\lambda_{j,0}^t-\lambda_{j,v}^t) e_{g_{i,u}(a(\bar{j},v))}& \textrm{if~}a_{\bar{j}}=0,\\
\hspace{2cm}+ \lambda_{j,0}^t e_{g_{i,u}(a)},&\\
 \lambda_{j,a_{\bar{j}}}^t e_{g_{i,u}(a)}, &\textrm{if~}a_{\bar{j}}\ne 0,
\end{array}\right.
\end{align*}
where $\hat{e}_a$ is defined in Section \ref{sec:notations}.
Since $\bar{j}<i$, then $(g_{i,u}(a))(\bar{j},v)=g_{i,u}(a(\bar{j},v))$ by \eqref{Eqn ga b (j,v)} and $(g_{i,u}(a))_{\bar{j}}=a_{\bar{j}}$ by \eqref{Eqn ga b}, therefore, $V_{i,u}[a:]A_{t,j}=B_{t,j,i}[a:]V_{i,u}$ for $a\in [0: N/w)$, i.e., $V_{i,u}A_{t,j}=B_{t,j,i}V_{i,u}$.

\item  If $\bar{j}>i$, by \eqref{Eqn_B>}, we have \eqref{Eqn_appA}.
\begin{figure*}
\begin{align}\label{Eqn_appA}
\nonumber B_{t,j,i}[a:]V_{i,u}&=\left\{\begin{array}{ll}
 (\lambda_{j,0}^t \hat{e}_{a}+\sum\limits_{v=1}^{w-1}   (\lambda_{j,0}^t-\lambda_{j,v}^t) \hat{e}_{a(\bar{j}-1,v)})\sum\limits_{b=0}^{N/w-1}\hat{e}_b^{\top} e_{g_{i,u}(b)},& \textrm{if~}a_{\bar{j}-1}=0,\\
 \lambda_{j,a_{\bar{j}-1}}^t \hat{e}_{a}\sum\limits_{b=0}^{N/w-1}\hat{e}_b^{\top} e_{g_{i,u}(b)}, &\textrm{if~}a_{\bar{j}-1}\ne 0,
\end{array}\right.\\
&=\left\{\begin{array}{ll}
 \lambda_{j,0}^t e_{g_{i,u}(a)}+\sum\limits_{v=1}^{w-1}   (\lambda_{j,0}^t-\lambda_{j,v}^t) e_{g_{i,u}(a(\bar{j}-1,v))},& \textrm{if~}a_{\bar{j}-1}=0,\\
 \lambda_{j,a_{\bar{j}-1}}^t e_{g_{i,u}(a)}, &\textrm{if~}a_{\bar{j}-1}\ne 0.
\end{array}\right.
\end{align}
\hrulefill
\end{figure*}
By  $\bar{j}>i$, we have $(g_{i,u}(a))(\bar{j},v)=g_{i,u}(a(\bar{j}-1,v))$ by \eqref{Eqn ga b (j,v)} and $(g_{i,u}(a))_{\bar{j}}=a_{\bar{j}-1}$ by \eqref{Eqn ga b}, therefore, $$V_{i,u}[a:]A_{t,j}=B_{t,j,i}[a:]V_{i,u}$$ for $a\in [0: N/w)$, i.e., $V_{i,u}A_{t,j}=B_{t,j,i}V_{i,u}$.
\end{itemize}

Combining the two cases above, we have
$S_{i}A_{t,j}=B_{t,j,i}R_{i,j}$ by \eqref{Eqn select matrix of code New1} and \eqref{Eqn repair matrix of code New1}. By \eqref{Eqn_Vt2}, we further have that $S_{i}A_{t,j}=B_{t,j,i}R_{i,j}$ holds for $i\in [0:n)$, $j\in [0: ms)$ with $i\not\equiv j \bmod m$.

\item [] Case 4b: For $i\in [0:n)$ and $j\in [ms:n)$ with $i\not\equiv j \bmod m$, we also have $S_{i}A_{t,j}=B_{t,j,i}R_{i,j}$, the proof is similar to that of Case 4a, therefore, we omit it here.
\end{itemize}

Collecting the above cases, \eqref{Eqn the first eq for proving Thm 9} is satisfied.

\section{Proof of Lemma \ref{le_BB}}\label{sec:pfBBcomute}
We only prove the case of $j,j'\in [0: i)$, the other cases can be proved similarly. 

i) For any $y,z\in [0: w)$, if $j\ne j'$, we have
\begin{align*}
&\hat{V}_{j,j',y,z}B_{v\overline{n}+j,i}B_{v'\overline{n}+j',i}\\=&x_{v\overline{n}+j}\lambda_{j}^{\delta(y\oplus1)}\hat{V}_{j,j',y\oplus1,z}B_{v'\overline{n}+j',i}\\
=&x_{v\overline{n}+j}x_{v'\overline{n}+j'}\lambda_{j}^{\delta(y\oplus1)}\lambda_{j'}^{\delta(z\oplus1)}\hat{V}_{j,j',y\oplus1,z\oplus1}
\end{align*}
and 
\begin{align*}
&\hat{V}_{j,j',y,z}B_{v'\overline{n}+j',i}B_{v\overline{n}+j,i}\\=&x_{v'\overline{n}+j'}\lambda_{j'}^{\delta(z\oplus1)}\hat{V}_{j,j',y,z\oplus1}B_{v\overline{n}+j,i}\\=&x_{v\overline{n}+j}x_{v'\overline{n}+j'}\lambda_{j'}^{\delta(z\oplus1)}\lambda_{j}^{\delta(y\oplus1)}\hat{V}_{j,j',y\oplus1,z\oplus1},
\end{align*}
which implies that $B_{v\overline{n}+j,i}B_{v'\overline{n}+j',i}=B_{v'\overline{n}+j',i}B_{v\overline{n}+j,i}$;
Otherwise, $j=j'$, we have
\begin{align*}
&\hat{V}_{j,y}B_{v\overline{n}+j,i}B_{v'\overline{n}+j',i}\\=&x_{v\overline{n}+j}\lambda_{j}^{\delta(y\oplus1)}\hat{V}_{j,y\oplus1}B_{v'\overline{n}+j,i}\\=&x_{v\overline{n}+j}x_{v'\overline{n}+j'}\lambda_{j}^{\delta(y\oplus1)+\delta(y\oplus2)}\hat{V}_{j,y\oplus2}
\end{align*}
and 
\begin{align*}
&\hat{V}_{j,y}B_{v'\overline{n}+j',i}B_{v\overline{n}+j,i}\\=&x_{v'\overline{n}+j'}\lambda_{j}^{\delta(y\oplus1)}\hat{V}_{j,y\oplus1}B_{v\overline{n}+j,i}\\=&x_{v\overline{n}+j}x_{v'\overline{n}+j'}\lambda_{j}^{\delta(y\oplus1)+\delta(y\oplus2)}\hat{V}_{j,y\oplus2},
\end{align*}
which also implies that $B_{v\overline{n}+j,i}B_{v'\overline{n}+j',i}=B_{v'\overline{n}+j',i}B_{v\overline{n}+j,i}$.

ii) If $j\ne j'$,  then $B_{v\overline{n}+j,i}-B_{v'\overline{n}+j',i}$ is non-singular if for any column vector $F$ of length $N/w$, $$(B_{v\overline{n}+j,i}-B_{v'\overline{n}+j',i})
F=\bf{0}$$ or equivalently 
 \begin{equation}\label{Eqn_BF=0}
\left(\begin{array}{c}
     \hat{V}_{j,j',0,0}  \\
      \hat{V}_{j,j',0,1}  \\
\vdots\\
\hat{V}_{j,j',0,w-1}  \\
\hat{V}_{j,j',1,0}  \\
\vdots\\
\hat{V}_{j,j',w-1,w-1}  \\
\end{array}\right)(B_{v\overline{n}+j,i}-B_{v'\overline{n}+j',i})
F=\bf{0},
\end{equation}
implies $F=\bf{0}$.

By \eqref{Eqn_B_le_i}, \eqref{Eqn_BF=0} is equivalent to
\begin{equation}
\left(\begin{array}{c}
   x_{v\overline{n}+j} \lambda_{j} \hat{V}_{j,j',1,0}-x_{v'\overline{n}+j'}\lambda_{j'}\hat{V}_{j,j',0,1}  \\
   x_{v\overline{n}+j} \lambda_{j} \hat{V}_{j,j',1,1}-x_{v'\overline{n}+j'}\hat{V}_{j,j',0,2}  \\
     \vdots\\
x_{v\overline{n}+j} \lambda_{j} \hat{V}_{j,j',1,w-1}-x_{v'\overline{n}+j'}\hat{V}_{j,j',0,0}\\
x_{v\overline{n}+j} \hat{V}_{j,j',2,0}-x_{v'\overline{n}+j'}\lambda_{j'}\hat{V}_{j,j',1,1}\\
\vdots\\
x_{v\overline{n}+j} \hat{V}_{j,j',0,w-1}-x_{v'\overline{n}+j'}\hat{V}_{j,j',w-1,0}
\end{array}\right)F=\bf{0}.
\end{equation} 

From the above equation, we have
\begin{equation}
 x_{v\overline{n}+j} \lambda_{j}^{\delta_{y\oplus1}} \hat{V}_{j,j',y\oplus1,z}F=x_{v'\overline{n}+j'}\lambda_{j'}^{\delta_{z\oplus1}}\hat{V}_{j,j',y,z\oplus1}F
\end{equation} 
for any $y,z\in [0:w)$, i.e.,
\begin{equation}
  \hat{V}_{j,j',y,z}F=\frac{x_{v'\overline{n}+j'}\lambda_{j'}^{\delta_{z\oplus1}}}{x_{v\overline{n}+j} \lambda_{j}^{\delta_y}}\hat{V}_{j,j',y\oplus(w-1),z\oplus1}F.
\end{equation} 
Recursively, we have
\begin{align*}
  \hat{V}_{j,j',y,z}F&=\frac{x_{v'\overline{n}+j'}\lambda_{j'}^{\delta_{z\oplus1}}}{x_{v\overline{n}+j} \lambda_{j}^{\delta_y}}\hat{V}_{j,j',y\oplus(w-1),z\oplus1}F\\
&=\frac{x_{v'\overline{n}+j'}^2\lambda_{j'}^{\delta_{z\oplus1}}\lambda_{j'}^{\delta_{z\oplus2}}}{x_{v\overline{n}+j}^2 \lambda_{j}^{\delta_y}\lambda_{j}^{\delta_{y\oplus(w-1)}}}\hat{V}_{j,j',y\oplus(w-2),z\oplus2}F\\
&=\frac{x_{v'\overline{n}+j'}^w\prod_{t=1}^{w}\lambda_{j'}^{\delta_{z\oplus t}}}{x_{v\overline{n}+j}^w \prod_{t=0}^{w-1}\lambda_{j}^{\delta_{y\oplus (w-t)}}}\hat{V}_{j,j',y,z}F\\
&=\frac{x_{v'\overline{n}+j'}^w\lambda_{j'}}{x_{v\overline{n}+j}^w \lambda_{j}}\hat{V}_{j,j',y,z}F
\end{align*}
by \eqref{Eqn_delta},
which together with i) implies $\hat{V}_{j,j',y,z}F=\bf{0}$ for $y,z\in [0:w)$, i.e., $F=\bf{0}$. 
That is, $B_{v\overline{n}+j,i}-B_{v'\overline{n}+j',i}$ is non-singular.

If $j=j'$ and $v\ne v'$, according to \eqref{Eqn_B_le_i}, $B_{v\overline{n}+j,i}-B_{v'\overline{n}+j',i}$ is non-singular since $ x_{v\overline{n}+j}\ne x_{v'\overline{n}+j}$ by ii). This finishes the proof.

\section{Proof of Lemma \ref{lem VAf}}\label{sec:pfC2}

We first prove that $V_{i,0}A_{v\overline{n}+j}^m=
B_{v\overline{n}+j,i}^mV_{i,0}$ for $j<i$.

For any $a\in [0:N/w)$, we have
\begin{align*}
&V_{i,0}[a:]A_{v\overline{n}+j}\\=&e_{g_{i,0}(a)}A_{v\overline{n}+j}\\=&x_{v\overline{n}+j}e_{g_{i,0}(a)}\bar{A}_{j}\\
=&\left\{
\begin{array}{ll}
x_{v\overline{n}+j}\lambda_je_{g_{i,0}(a)(j,g_{i,0}(a)_j\oplus1)} , &\mbox{~if~}g_{i,0}(a)_j=0,\\
x_{v\overline{n}+j}e_{g_{i,0}(a)(j,g_{i,0}(a)_j\oplus1)} ,&\mbox{~if~}g_{i,0}(a)_j\ne 0,
\end{array}
\right. \\=&\left\{
\begin{array}{ll}
x_{v\overline{n}+j}\lambda_je_{g_{i,0}(a(j,a_j\oplus1))} , &\mbox{~if~}a_j=0,\\
x_{v\overline{n}+j}e_{g_{i,0}(a(j,a_j\oplus1))} ,&\mbox{~if~}a_j\ne 0,
\end{array}
\right. 
\end{align*}
where the second equality follows from \eqref{Eqn Zigzag d A}, the third equality follows from \eqref{Eqn_YB2_CodingM1}, and the fourth equality holds according to \eqref{Eqn ga b} and \eqref{Eqn ga b (j,v)}.

In addition, by \eqref{Eqn_B_le_i}, we have
\begin{align*}
B_{v\overline{n}+j,i}[a:]V_{i,0}&=\hat{e}_aB_{v\overline{n}+j,i}V_{i,0}\\
&=\left\{
\begin{array}{ll}
x_{v\overline{n}+j}\lambda_j\hat{e}_{a(j, a_j\oplus1)}V_{i,0}, &\mbox{~if~}a_j=0,\\
x_{v\overline{n}+j}\hat{e}_{a(j,a_j\oplus1)}V_{i,0} ,&\mbox{~if~}a_j\ne 0,
\end{array}
\right. \\&=   \left\{
\begin{array}{ll}
x_{v\overline{n}+j}\lambda_je_{g_{i,0}(a(j, a_j\oplus1))}, &\mbox{~if~}a_j=0,\\
x_{v\overline{n}+j}e_{g_{i,0}(a(j,a_j\oplus1))} ,&\mbox{~if~}a_j\ne 0.
\end{array}
\right. 
\end{align*}
The above two equations imply that $V_{i,0}[a:]A_{v\overline{n}+j}=B_{v\overline{n}+j,i}[a:]V_{i,0}$, then $V_{i,0}A_{v\overline{n}+j}=B_{v\overline{n}+j,i}V_{i,0}$.
Recursively, we have $V_{i,0}A_{v\overline{n}+j}^m=B_{v\overline{n}+j,i}^mV_{i,0}$ for $m\in [0: r-w)$. 

For $j>i$ and for any $a\in [0:N/w)$, we have
\begin{align*}
&V_{i,0}[a:]A_{v\overline{n}+j}\\=&x_{v\overline{n}+j}e_{g_{i,0}(a)}\bar{A}_{j}\\=&\left\{
\begin{array}{ll}
x_{v\overline{n}+j}\lambda_je_{g_{i,0}(a)(j,g_{i,0}(a)_j\oplus1)} , &\mbox{~if~}g_{i,0}(a)_j=0,\\
x_{v\overline{n}+j}e_{g_{i,0}(a)(j,g_{i,0}(a)_j\oplus1)} ,&\mbox{~if~}g_{i,0}(a)_j\ne 0,
\end{array}
\right. 
\\=&\left\{
\begin{array}{ll}
x_{v\overline{n}+j}\lambda_je_{g_{i,0}(a(j-1,a_{j-1}\oplus1))} , &\mbox{~if~}a_{j-1}=0,\\
x_{v\overline{n}+j}e_{g_{i,0}(a(j-1,a_{j-1}\oplus1))} ,&\mbox{~if~}a_{j-1}\ne 0,
\end{array}
\right.
\end{align*}
where the third equality holds according to \eqref{Eqn ga b} and \eqref{Eqn ga b (j,v)}.

In addition, by \eqref{Eqn_B_gt_i}, we have
\begin{align*}
&B_{v\overline{n}+j,i}[a:]V_{i,0}\\=&\left\{
\begin{array}{ll}
x_{v\overline{n}+j}\lambda_j\hat{e}_{a(j-1, a_{j-1}\oplus1)}V_{i,0}, &\mbox{~if~}a_{j-1}=0,\\
x_{v\overline{n}+j}\hat{e}_{a(j-1,a_{j-1}\oplus1)}V_{i,0} ,&\mbox{~if~}a_{j-1}\ne 0.
\end{array}
\right. \\=&   \left\{
\begin{array}{ll}
x_{v\overline{n}+j}\lambda_je_{g_{i,0}(a(j-1,a_{j-1}\oplus1))}, &\mbox{~if~}a_{j-1}=0,\\
x_{v\overline{n}+j}e_{g_{i,0}(a(j-1,a_{j-1}\oplus1))} ,&\mbox{~if~}a_{j-1}\ne 0.
\end{array}
\right. 
\end{align*}
The above two equations imply that $$V_{i,0}[a:]A_{v\overline{n}+j}=B_{v\overline{n}+j,i}[a:]V_{i,0},$$ therefore, we similarly have
\begin{equation*}
V_{i,0}A_{v\overline{n}+j}^m=B_{v\overline{n}+j,i}^mV_{i,0} \mbox{~for~} m\in [0: r-w).    
\end{equation*}

\section*{Acknowledgment}
The authors would like to thank the Associate Editor
Prof. Farzad Farnoud and the two anonymous reviewers for
their valuable suggestions and comments, which have greatly
improved the presentation and quality of this paper.


\begin{IEEEbiographynophoto}{Jie Li} (Senior~Member, IEEE) received the B.S. and M.S. degrees in mathematics from Hubei University, Wuhan, China, in 2009 and 2012, respectively, and received the Ph.D. degree from the Department of Communication Engineering, Southwest Jiaotong University, Chengdu, China, in 2017.

From  2015 to   2016, he was a visiting Ph.D. student at the Department of Electrical Engineering and Computer Science, The University of Tennessee at Knoxville, TN, USA.  From   2017 to 2019, he was a postdoctoral researcher at the Department of Mathematics, Hubei University, Wuhan, China. From 2019 to   2021, he was a postdoctoral researcher at the Department of Mathematics and Systems Analysis, Aalto University, Finland. From 2021 to 2026, he was a Senior Researcher with the Theory Lab, Central Research Institute, 2012 Labs, Huawei Technologies Co., Ltd., Hong Kong SAR, China.  He is currently a Research Fellow with the Department of Computer Science, City University of Hong Kong, Hong Kong SAR, China. His research interests include coding for distributed storage and distributed computing, private information retrieval, and sequence design.

Dr. Li received the IEEE Jack Keil Wolf ISIT Student Paper Award and the  
Honor of Outstanding Graduates of Sichuan Province both in 2017.
\end{IEEEbiographynophoto}

\begin{IEEEbiographynophoto}{Yi Liu} received the B.S. degree in mathematics and applied mathematics from Xihua University, Chengdu, China, in 2014, and received the Ph.D. degree from the department of information security, Southwest Jiaotong University, Chengdu, China, in 2022. He is currently senior engineer with Huawei Tech. Investment Co., Limited, Chengdu, China. His research interest includes coding for distributed storage and private information retrieval.
\end{IEEEbiographynophoto}

\begin{IEEEbiographynophoto}{Xiaohu Tang} (Fellow, IEEE)  received the B.S. degree in applied mathematics
from Northwest Polytechnic University, Xi’an, China, in 1992, the M.S. degree in applied mathematics from Sichuan University, Chengdu, China, in 1995, and the Ph.D. degree in electronic engineering from Southwest Jiaotong
University, Chengdu, in 2001.

From 2003 to 2004, he was a Research Associate with the Department of
Electrical and Electronic Engineering, The Hong Kong University of Science
and Technology. From 2007 to 2008, he was a Visiting Professor with the
University of Ulm, Germany. Since 2001, he has been with the School of
Information Science and Technology, Southwest Jiaotong University, where he
is currently a Professor. His research interests include coding theory, network
security, distributed storage, and information processing for big data.

Dr. Tang was the recipient of the National Excellent Doctoral Dissertation
award in 2003 (China), the Humboldt Research Fellowship in 2007
(Germany), and the Outstanding Young Scientist Award by NSFC in 2013
(China). He served as an Associate Editor for several journals including \textit{IEEE
Transactions on Information Theory} and \textit{IEICE Transactions on
Fundamentals}, and served on a number of technical program committees of
conferences.
\end{IEEEbiographynophoto}

\begin{IEEEbiographynophoto}{Yunghsiang~S.~Han}(Fellow, IEEE) was born in Taipei, Taiwan, in 1962. He earned his B.Eng. and M.Eng. degrees in electrical engineering from National Tsing Hua University, Hsinchu, Taiwan, in 1984 and 1986, respectively, and his Ph.D. from Syracuse University's School of Computer and Information Science, Syracuse, NY, USA, in 1993. From 1986 to 1988, he was a Lecturer at the Ming-Hsin Engineering College, Hsinchu. He served as a Teaching Assistant from 1989 to 1992 and as a Research Associate with the School of Computer and Information Science at Syracuse University from 1992 to 1993. From 1993 to 1997, he was an Associate Professor with the Department of Electronic Engineering, Hua Fan College of Humanities and Technology, Taipei Hsien, Taiwan. He served in the Department of Computer Science and Information Engineering at National Chi Nan University, Nantou, Taiwan, from 1997 to 2004. He was promoted to Professor in 1998. He was a Visiting Scholar with the Department of Electrical Engineering, University of Hawaii at Manoa, Honolulu, HI, USA, from June 2001 to October 2001; a SUPRIA Visiting Research Scholar with the Department of Electrical Engineering and Computer Science and the CASE Center, Syracuse University, from September 2002 to January 2004 and again from July 2012 to June 2013; and a Visiting Scholar with the Department of Electrical and Computer Engineering, The University of Texas at Austin, Austin, TX, USA, from August 2008 to June 2009. He was with the Graduate Institute of Communication Engineering, National Taipei University, Taipei, from August 2004 to July 2010. From August 2010 to January 2017, he served as a Chair Professor in the Department of Electrical Engineering at National Taiwan University of Science and Technology. He has been a Chair Professor at the National Taipei University since February 2015. From February 2017 to February 2021, he was with the School of Electrical Engineering and Intelligentization, Dongguan University of Technology, China. From June 2021 to May 2026, he was with the Shenzhen Institute for Advanced Study, University of Electronic Science and Technology of China. He is currently with the School of Computing and Information Technology, Great Bay University. His research interests include error-control coding, wireless networks, and security. He was a winner of the 1994 Syracuse University Doctoral Prize. One of his papers won the prestigious 2013 ACM CCS Test-of-Time Award in Cybersecurity.
\end{IEEEbiographynophoto}

\begin{IEEEbiographynophoto}
{Bo Bai} (Senior Member, IEEE) received the B.S. 
degree (Hons.) from the School of Communication Engineering, Xidian University, Xi’an, China, 
in 2004, and the Ph.D. degree from the Department
of Electronic Engineering, Tsinghua University, Beijing, China, in 2010. 

He was a Research Assistant and a Research 
Associate with the Department of Electronic and 
Computer Engineering, The Hong Kong University of Science and Technology, from April 2009 to 
September 2010 and October 2010 to April 2012,  
respectively. From July 2012 to January 2017, he was an Assistant Professor  
with the Department of Electronic Engineering, Tsinghua University. He has  
obtained the support from the Backbone Talents Supporting Project of  
Tsinghua University. Currently, he is an Information Theory Scientist and the  
Director of the Theory Laboratory, 2012 Labs, Huawei Technologies Company  
Ltd., Hong Kong. He has authored more than 130 papers in major IEEE and  
ACM journals and conferences, two book chapters, and one textbook. His 
research interests include semantic information theory, B5G/6G networking,  
and graph informatics. He served as a member for the IEEE ComSoc WTC  
Committee and the IEEE ComSoc SPCE Committee. He also served as a  
TPC Member for several IEEE conferences, such as ICC, Globecom, WCNC,  
VTC, and ICCC. He was a recipient of the Student Travel Grant at IEEE  
Globecom 2009 and the Best Paper Award at the IEEE ICC 2016. He was  
invited as a Young Scientist Speaker at the IEEE TTM 2011. He received the  
Honor of Outstanding Graduates of Shaanxi Province and the Honor of Young 
Academic Talent of Electronic Engineering in Tsinghua University. He is one  
of the Founding Vice Chair of the IEEE TCCN SIG on Social Behavior  
Driven Cognitive Radio Networks. He served as the TPC Co-Chair for the 
IEEE Infocom 2018—1st AoI Workshop and the IEEE Infocom 2019—2nd  
AoI Workshop, the TPC Co-Chair for the IEEE ICCC 2018, and an Industrial  
Forum and the Exhibition Co-Chair for the IEEE HotICN 2018. He also served  
as  an
Associate Editor of \textit{IEEE COMMUNICATIONS SURVEYS AND TUTORIALS} and a reviewer for several major IEEE and ACM journals and conferences.
\end{IEEEbiographynophoto}

\begin{IEEEbiographynophoto}
{Gong Zhang} is a chief researcher scientist in future network architecture of
Huawei 2012 Labs. He has over 20 years of research experience as a system
architect in Networks, distributed systems, and communication systems. He had
contributed more than 90 patents globally. He was a Senior Researcher,
leading future internet research and cooperative communication, and did
Mobility research program since 2005. In 2009, Principal Researcher, in
charge of advanced network technology research department, leading
research on future networks, distributed computing, database systems, and
data analysis. He proposed a stream-based research system in the network to
maintain the network and find added value for carriers. He was a Principal
Researcher, leading system group in data mining and machine learning since
2012. Recently, the major research directions are Network
architecture and large-scale distributed systems.
\end{IEEEbiographynophoto}
\end{document}